\documentclass[letterpaper,10pt]{IEEEtran}
\usepackage[T1]{fontenc}
\usepackage[latin1]{inputenc}
\usepackage{float}
\usepackage{amsmath}
\usepackage[noadjust]{cite}
\usepackage{amssymb}
\usepackage{enumerate}
\usepackage{amsthm}
\usepackage{graphics,epsfig,psfrag}
\usepackage{subfigure,epsf,pstricks,subfloat}
\usepackage{bm}
\usepackage{caption}
\PassOptionsToPackage{hyphens}{url}
\usepackage{hyperref}
\usepackage{breakurl}
\usepackage{comment}
\usepackage{psfrag}
\floatstyle{ruled}
\newfloat{algorithm}{tbp}{loa}
\floatname{algorithm}{Routine}

\usepackage{algorithmic}
\algsetup{linenodelimiter= }

\def\bfI{{\mathbf I}}

\def\bfV{{\mathbf V}}

\newcommand{\lsem}{[\![} \newcommand{\rsem}{]\!]}

\def\V{{\bf V}}

\def\G{{\mathfrak{G}}}
\def\P{{\mathcal{P}}}
\def\I{{\mathcal{I}}}
\def\J{{\mathcal{J}}}

\newcommand{\C}{\mathbb{C}}

\newcommand{\R}{\mathbb{R}}

\title{Identification of Successive ``Unobservable''  Cyber Data Attacks in Power Systems Through Matrix Decomposition} 

\author{
\IEEEauthorblockN{Pengzhi Gao, \textit{Student Member, IEEE}, Meng Wang, \textit{Member, IEEE,}   Joe H. Chow,  \textit{Fellow, IEEE,}  \\   Scott G. Ghiocel, \textit{Member, IEEE,} Bruce Fardanesh, \textit{Fellow, IEEE,} \\ George Stefopoulos, \textit{Member, IEEE}, and Michael P. Razanousky
}
 \thanks{P. Gao, M. Wang and  J. H. Chow are with the Dept. of Electrical, Computer, and Systems Engineering,   Rensselaer Polytechnic Institute, Troy, NY.  Email:   \{gaop, wangm7, chowj\}@rpi.edu.
 }
 \thanks{S. G. Ghiocel is with Exponent, New York, NY. Email: sghiocel@exponent.com.
 }
 \thanks{B. Fardanesh, and G. Stefopoulos are with New York Power Authority, White Plains, NY. Email: \{Bruce.Fardanesh, George.Stefopoulos\}@nypa.gov.
 }
 \thanks{M. P. Razanousky is with New York State Energy Research and Development Authority, Albany, NY. Email: mpr@nyserda.org.
 }
 \thanks{Partial and preliminary results have appeared in \cite{WGGCFSR14}.}
}

\usepackage{graphicx}
\usepackage{epstopdf}
\newtheorem{theorem}{Theorem}
\newtheorem{lemma}{Lemma}

\begin{document} 

\maketitle \thispagestyle{empty} \pagestyle{empty}

%%%%%%%%%%%%%%%%%%%%%%%%%%%%%%%%%%%%%%%%%%%%%%%%%%%%%%%%%%%%%%%%%%%%%%%%%%%%%%%%
\begin{abstract}
This paper presents a new framework of identifying a series of cyber data attacks on power system synchrophasor measurements. We focus on detecting ``unobservable'' cyber data attacks that cannot be detected by any existing method that purely relies on measurements received at one time instant. Leveraging the approximate low-rank property of phasor measurement unit (PMU) data, we formulate the identification problem of successive unobservable cyber attacks as a matrix decomposition problem of a low-rank matrix plus a transformed column-sparse matrix. We propose a convex-optimization-based method and provide its theoretical guarantee in the data identification. Numerical experiments on actual PMU data from the Central New York power system and synthetic data are conducted to verify the effectiveness of the proposed method.    
\end{abstract}
\noindent\begin{IEEEkeywords}
cyber data attacks, low-rank matrix,  matrix decomposition, synchrophasor measurements.
\end{IEEEkeywords}
%%%%%%%%%%%%%%%%%%%%%%%%%%%%%%%%%%%%%%%%%%%%%%%%%%%%%%%%%%%%%%%%%%%%%%%%%%%%%%%%

\section{Introduction}\label{sec:intro}

\IEEEPARstart{T}{he} integration of cyber infrastructures into future smart grids greatly enhances the monitoring, dispatch, and scheduling of power systems. Such integration, however, makes the power systems more susceptible to cyber attacks. It is reported that cyber spies have penetrated U.S. electrical grid \cite{meserve07}. Researchers  have also launched an experimental cyber attack that caused a generator to self-destruct \cite{HLEO09}. 
 
State estimation \cite{AE04} is a critical component of power system monitoring. System state is estimated based on the obtained measurements across the system. Bad data can affect the state estimation and mislead the system operator. Many efforts have been devoted to develop methods that can identify   bad data, see e.g., \cite{CA06,HSKF75,MG83,MS71,XWCT13}. 
  
Cyber data attacks (firstly studied in \cite{LNR11}) can be viewed as ``the worst interacting bad data injected by an adversary''\cite{KJTT10}. Malicious intruders with system configuration information can simultaneously manipulate   multiple  measurements so that these attacks cannot be detected by any bad data detector. Because the removal of affected measurements would make the system unobservable, these attacks are termed as ``unobservable attacks''\footnote{The term ``unobservable'' is used in this sense throughout the paper.} in \cite{KJTT10}.  

State estimation in the presence of cyber data  attacks has attracted much research attention recently  \cite{BRWKNO10,DS10,KJTT10,LEDEH14,LNR11,SJ13,STJ10,TKPC11}. Existing approaches include protecting a small number of key measurement units such that the intruders cannot inject unobservable attacks without hacking protected units  \cite{BRWKNO10,DS10,KP11}, as well as detectors designed for attacks in the observable regime \cite{KJTT10}. The research on the detection of unobservable attacks is still limited. Refs. \cite{SJ13,LEDEH14} proposed different methods to detect unobservable attacks in Supervisory Control and Data Acquisition (SCADA) system. The method in \cite{SJ13} relies critically on the assumption that the measurements at different time instants are i.i.d. samples of random variables. This assumption might not hold when the system is under disturbance. Ref. \cite{LEDEH14} focused on the scenarios that an intruder attacks a different set of measurements at each time instant, and no theoretical analysis of the detection performance is provided in  \cite{LEDEH14}.

This paper considers cyber data attacks to PMU measurements. It focuses on the case when an intruder injects unobservable data attacks to the same set of PMUs constantly. Because PMUs under attack do not provide any accurate measurement at any time instant to the operator, the attack identification in this case is very challenging and has not been addressed before. We propose a method that can identify the successive unobservable cyber data attacks and provide the theoretical guarantee even when the system is under disturbance. The intuition is that even though an intruder can constantly inject data attacks that are consistent with each other at each time instant, as long  as   the intruder does not know the system dynamics, one can identify the  attacks by comparing time series of different PMUs and locating the PMUs that exhibit abnormal dynamics.

Because PMU measurements are synchronized and correlated, the high-dimensional PMU data matrix exhibits low-rank property \cite{CXK13,DKM12,GWGC14,WGGCFSR14}. We formulate the identification problem as a matrix decomposition problem of a low-rank matrix plus a transformed column-sparse matrix. The matrix decomposition problem has attracted much research attention recently, see e.g., \cite{CLMW11,CSPW11,RFP10,XCS12}, and have wide applications in areas like Internet monitoring \cite{LCD04,MMG13,TJ03}, medical imaging \cite{FNDRL06,GCSZ11}, and image processing \cite{BJ03}. The situation that one component is a transformed column-sparse matrix, however, has not been  addressed before.

The contributions of this paper are threefold. (1) We  propose the idea of  exploiting   spatial-temporal correlations in PMU measurements to identify unobservable data attacks. (2) We formulate the identification problem into a matrix decomposition problem and propose a computationally efficient method that does not require the modeling of power system dynamics. (3) We provide theoretical guarantees of attack detection, as well as the general matrix decomposition problem. 

The rest of the paper is organized as follows. We formulate our problem and point out its connection to other applications in Section \ref{sec:model}. We describe our detection method and analyze its theoretical guarantee with both noiseless  (Section \ref{sec:method}) and noisy measurements (Section \ref{sec:method1}). Section \ref{sec:simu} records our numerical experiments. We conclude the paper in Section \ref{sec:con}. 
\section{Problem Formulation and Related Work} \label{sec:model}

\subsection{Low-rankness of PMU measurements}\label{sec:pmu}

Consider a $n$-bus power grid with PMUs installed on some buses. Let $p$ denote the total number of PMU channels that measure bus voltage   and line current phasors\footnote{In  three phase AC  systems, a phasor is defined as a complex number that represents both the magnitude and phase angle of the sinusoidal waveforms.}. Phasors are expressed in Cartesian coordinates throughout the paper. Matrix $M \in \C^{t\times p}$ contains the collected phasor measurements in $t$  synchronized time instants. $\bar{\J} \in \lsem p \rsem$ denotes the set of PMU channels that are under data attacks. The observed measurement matrix can be presented as 
\begin{equation}\label{eqn:M}
M=\bar{L}+\bar{D}+N, 
\end{equation}
where $\bar{L} \in \C^{t\times p}$ represents the actual phasors without data attacks, $\bar{D}\in \C^{t\times p}$ represents the additive errors introduced by an intruder, and $N$ represents the measurement noise. 

High-dimensional PMU data matrices exhibit low-rank property \cite{CXK13,DKM12,GWGC14,WGGCFSR14}. We analyzed actual PMU data from six multi-channel PMUs deployed in the Central New York (NY) Power System (Fig.~\ref{fig:NY}). Six PMUs measure twenty-three voltage and current phasors, and the data rate is  thirty samples per second per channel. 
Fig.~\ref{fig:current} shows the current magnitudes of PMU data in twenty seconds. An event occurs around $2.5$s.   The obtained data are collected into a $600 \times 23$ matrix.  Fig.~\ref{fig:svd} plots the singular values of the  
matrix  with the ten largest ones being 832.8, 194.8, 35.1, 18.1, 4.3, 2.5, 2.1, 1.3, 1.2, 0.5. Therefore, we can approximate the  $600 \times 23$ matrix by a low-rank matrix with little approximation error.  

\begin{figure}[h] 
\begin{center}
\includegraphics[height=0.43 \linewidth]{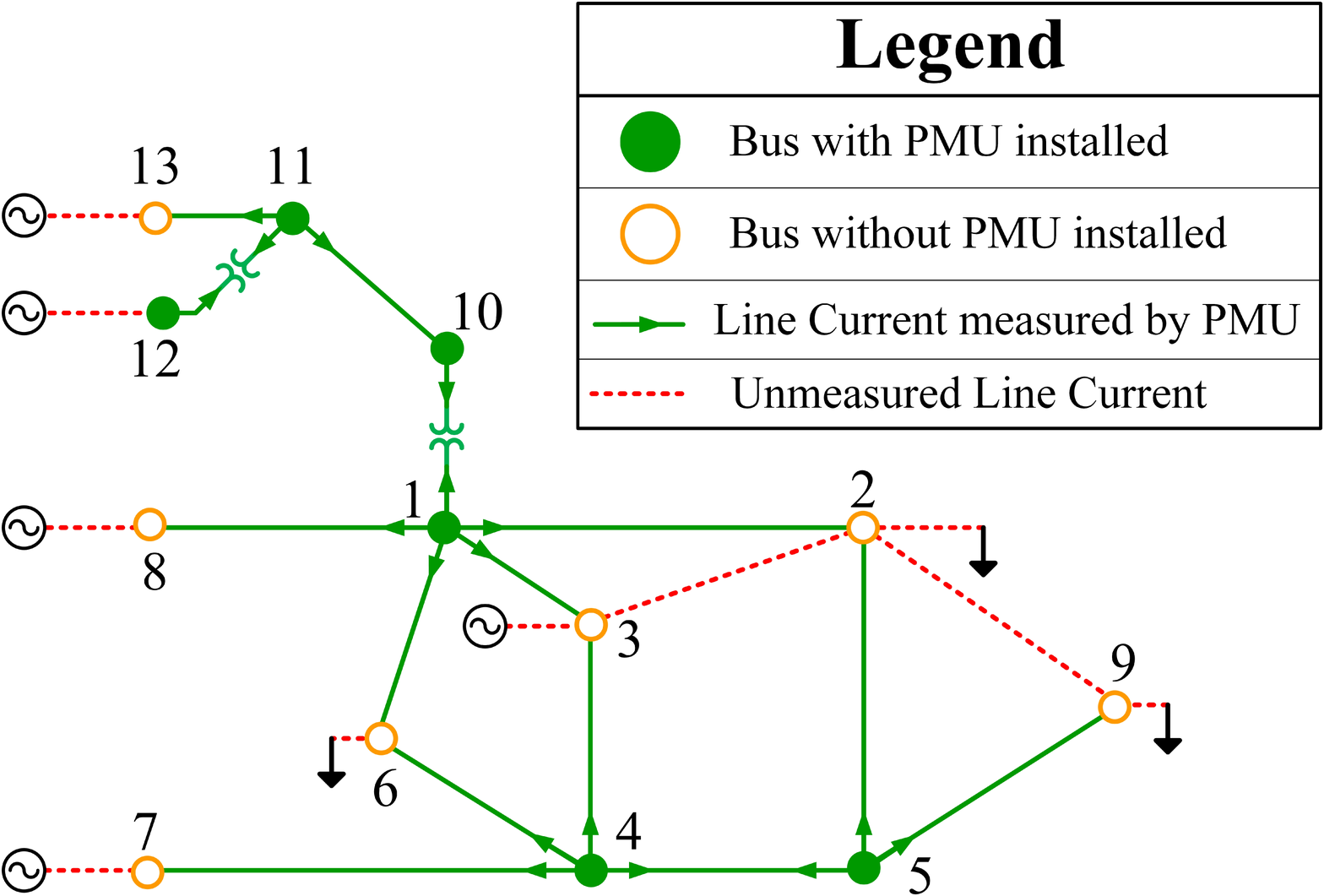}
\caption{PMUs in the Central NY Power System. (Circles and lines represent buses and transmission lines. A PMU measures the   voltage phasor and the incident   current phasors of the bus where it is located.)} \label{fig:NY}
\end{center}
\end{figure} 

\begin{figure}[h]
\begin{center}
\psfrag{Current Magnitude}[][][0.7]{Current magnitude}
\psfrag{Time(s)}[][][0.7]{Time(s)}
\vspace{-5mm} 
\includegraphics[height=0.28\linewidth]{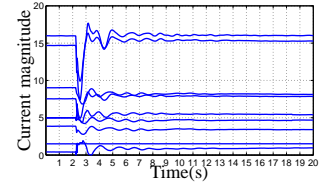} 
\caption{Visualization of Partial PMU data (Magnitude of nine current phasors)} \label{fig:current}\vspace{-5mm} 
\end{center}
\end{figure}

\begin{figure}[t] 
\begin{center}
\psfrag{Singular value index}[][][0.7]{Singular value index}
\psfrag{Singular values}[][][0.7]{Singular values}
\includegraphics[height=0.28\linewidth]{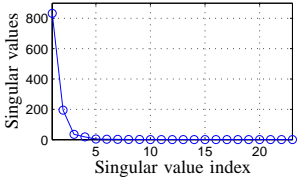}
\caption{Singular values of PMU data matrix in decreasing order} \label{fig:svd} 
\end{center}
\end{figure}

The Singular Value Decomposition (SVD) of $\bar{L}$ is
\begin{equation} 
\bar{L}=\bar{U}\bar{\Sigma}\bar{V}^ {\dagger},
\end{equation}
where $\bar{U}\in \C^{t\times r}$,   $\bar{\Sigma}\in \C^{r\times r}$,  $\bar{V}\in \C^{p \times r}$  ($r \ll t, p$). We assume throughout the paper that   nonzero columns of $\bar{D}$ do  not lie in the column space of $\bar{L}$ ($\bar{D} \neq \bar{U} \bar{U}^\dag \bar{D}$). It is a legitimate assumption when the intruders do not have full information about the system dynamics. 
The notations are summarized in Table \ref{tbl:notation}. Matrix $A$ is \textit{column-sparse} if it contains a small fraction of non-zero columns. We call the set of indices of nonzero columns the \textit{column support} of $A$.

\begin{table}[h]
\caption{Notations}
\begin{tabular}{|p{0.55in} p{2.71in}|}
\hline $A_i$, $A_{i,:}$ & the $i$th column  and the $i$th row of matrix $A$, respectively. \\
\hline $A_{\I}$ & the submatrix of $A$ with column indices in set $\I$. \\
\hline $A^{\ddagger}$, $A^{\dagger}$ & the conjugate and conjugate transpose matrix of $A$.\\
\hline $\P_{\I}(A)$ & matrix obtained from $A$ by setting $A_i$ to zero for all $i \notin \I$.  \\ 
\hline  $A \in \P_{\I}$ &  if and only if $\P_{\I}(A)=A$.\\
\hline  $\|A\|$, $\|A\|_{F}$ & the spectral and Frobenius norm of $A$, respectively. \\
\hline  $\|A\|_{*}$  &  the nuclear norm of $A$, which is the sum of singular values. \\ 
\hline $\|A\|_{1,2}$ &  the sum of $\ell_2$ norms of the columns of $A$. \\
\hline  $\|A\|_{\infty, 2}$ & the largest $\ell_2$ norm of the columns. \\
\hline $\P_{U} (A)$ & $:= U U^\dag A$, the projection of  $A$ onto the column space of $L$. \\
\hline  $\P_{V} (A)$ & $:= A V V^\dag$, the projection of  $A$ onto the row space.  \\
\hline  $\P_{T} (\cdot)$ & $:=\P_{U} (\cdot)+\P_{V} (\cdot) - \P_{U}\P_{V}(\cdot)$. \\
\hline   $\P_{U^\perp} (A)$ & $:= (I-U U^\dag)A$. \\
\hline $\P_{V^\perp} (A)$ & $:= A(I- V V^\dag)$. \\
\hline  $\P_{T^\perp} (A)$ & $:= \P_{U^\perp}\P_{V^\perp} (A)$.\\
\hline $A \in \P_{T}$ & if and only if $\P_{T}(A)=A$.\\
\hline $\I^c$ &   the complimentary set of set $\I$. \\
\hline
\end{tabular} 
\label{tbl:notation}
\end{table}

\subsection{Unobservable cyber data attacks and problem formulation}

We use bus voltage phasors as state variables, and let $X \in \C^{t\times n}$ contain the state variables at $t$ instants. We use the $\pi$ equivalent model to represent a transmission line (Fig.~\ref{fig:pi}). $Z^{ij}$ and $Y^{ij}$ denote the impedance and admittance of the transmission line between bus $i$ and bus $j$. Current $I^{ij}$ from bus $i$ to bus $j$ is related to bus voltage $V^i$ and $V^j$ by
\begin{equation} 
 I^{ij}=\frac{V^i- V^j}{Z^{ij}}+V^i \frac{Y^{ij}}{2}.
\end{equation}

\begin{figure}[h] 
\begin{center}
\includegraphics[height=0.2\linewidth]{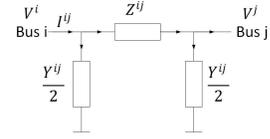}
\caption{$\pi$ model of a transmission line} \label{fig:pi} \vspace{-3 mm} 
\end{center}
\end{figure} 

We define $\bar{W} \in \C^{p \times n}$ as follows. If the $k$th PMU channel measures the voltage phasor of bus $j$,  $\bar{W}_{kj}=1$; if it measures the current phasor from bus $i$ to bus $j$, then $\bar{W}_{ki}=1/Z^{ij}+Y^{ij}/2$, $\bar{W}_{kj}=-1/Z^{ij}$; $\bar{W}_{kj}=0$ otherwise. The PMU measurements and the state variables are related by 
\begin{equation}\label{eqn:linear}
\bar{L}= X \bar{W}^T.
\end{equation}
The attack at time $t$, denoted by data injection $\bar{D}_{t,:}$, is called \textit{unobservable}\footnote{\cite{LNR11}  focuses on DC model where power measurements and state variables are approximately related by linear equations. Here PMU measurements and state variables are accurately related by linear equation (\ref{eqn:linear}).} if and only if
\begin{align}
& \bar{D}_{t,:}=c^t\bar{W}^T
\end{align}
holds for some nonzero row vector $c^t \in \C^{1\times n}$. In this case, 
\begin{align}
&\bar{L}_{t,:}+\bar{D}_{t,:}=(X_{t,:}+ c^t) \bar{W}^T,
\end{align}
and the operator would have the wrong impression that the state is $X_{t,:}+ c^t$. We focus on the cases that the attacks from time 1 to $t$ are all unobservable\footnote{ Our detection method can be extended to cases that both unobservable  and observable attacks exist. See the beginning of Section \ref{sec:analysis}}, then we have
\begin{equation} \label{eqn:state}
\bar{D}=\left[\begin{array}{c}
c^1\\
\vdots\\
c^t
\end{array}\right]\bar{W}^T:= \bar{C} W^T,
\end{equation}
where $W_j=\bar{W}_j/\|\bar{W}_j\|$. $\bar{C}$ represents the additive error (up to a scaling factor) to bus voltages due to data attacks, i.e., $\|\bar{W}_j\|\bar{C}_j$ is the error to bus voltage $V^j$. Let $\bar{\I} \in \lsem n \rsem$ denote the column support of $\bar{C}$. We assume $\bar{C}$ is column-sparse because intruders might only alter some of the state variables due to resource constraints. With increasing installation of PMUs, we anticipate that the total number of PMU channels $p$ will be  larger than the number of buses $n$. The transform in (\ref{eqn:state}) reduces the degree of freedom in $\bar{D}$. Combining (\ref{eqn:M}) and (\ref{eqn:state}), the obtained measurements under attack can be written as
\begin{equation}\label{eqn:problem}
M=\bar{L}+\bar{C}W^T+N.
\end{equation}

The attack identification problem is formulated as follows. Given $M$ and $W$, is it possible to separate $\bar{L}$ and $\bar{C}$? We assume noise level is bounded and given, i.e., $\|N\|_{F}\leq \eta$. We say a method can \textit{identify} an unobservable attack if it successfully determines the set of PMU channels that are under attack and recovers measurements that are not attacked. 

Although cannot be detected at a given time instant, the unobservable attacks can be detected if the time series in the affected PMU channels exhibit dynamics different from those of unaffected PMUs. Mathematically, the matrix decomposition is possible if columns in $\bar{D}$ do  not belong to the column space of $\bar{L}$. 

\begin{figure}[h] 
\begin{center}
\includegraphics[height=0.45\linewidth]{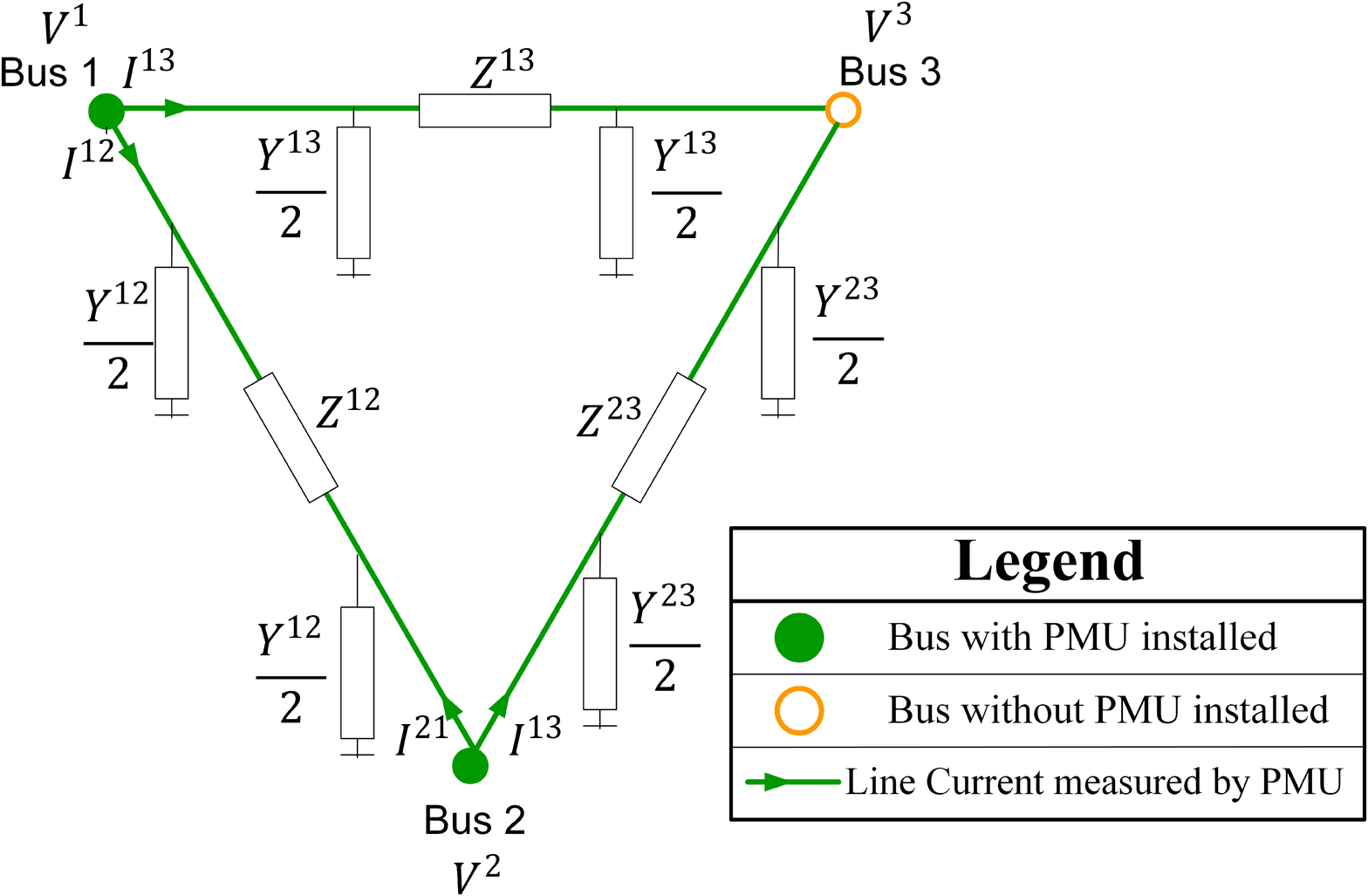}
\caption{Three-bus example. PMUs are installed at bus 1 and bus 2 measuring the corresponding voltage phasors and incident line current phasors. } \label{fig:example} \vspace{-5mm} 
\end{center}
\end{figure} 

We use a three-bus network (Fig.~\ref{fig:example}) to illustrate the notations. Let  $\V^i$ and $\bfI^{ij}$ ($i,j\in\{1,2,3\}$) in $\C^{t \times 1}$ denote the bus voltages and line currents in $t$ instants. Then
\begin{align}
\bar{L} &= [\bfV^1 \  \bfI^{12} \  \bfI^{13} \ \bfV^2 \  \bfI^{21} \ \bfI^{23}]= [\V^1 \  \V^2 \  \V^3]    \bar{W}^T
\end{align}
where $\bar{W}^T$ is
\begin{equation}\nonumber
%\bar{W}^T=
\notag \left[  \begin{array}{ccccccc}
1 & \frac{1}{Z^{12}}+\frac{Y^{12}}{2} & \frac{1}{Z^{13}}+\frac{Y^{13}}{2} & 0 & -\frac{1}{Z^{12}}  & 0\\
0 & -\frac{1}{Z^{12}} & 0 & 1 & \frac{1}{Z^{12}}+\frac{Y^{12}}{2}&  \frac{1}{Z^{23}}+\frac{Y^{23}}{2}\\
0 &  0 & -\frac{1}{Z^{13}} & 0 & 0 & -\frac{1}{Z^{23}}
\end{array}\right].
\end{equation}

Suppose the intruder manipulates measurements in all channels of PMU 1 and the channel of PMU 2 that measures $\bfI^{21}$ and $\bfI^{23}$  so that the system operator would have the wrong impression that the system states are  $[\V^1  +\bm\beta^1 \  \V^2  \  \V^3+\bm\beta^2]$ for any nonzero $\bm\beta^1$, $\bm\beta^2 \in \C^{t\times 1}$. In this case, the observed measurements under attacks when there is no noise  are
\begin{align}
&M = [\V^1  +\bm\beta^1 \ \  \V^2  \ \  \V^3+\bm\beta^2]    \bar{W}^T \nonumber \\
& = [\V^1 +\bm\beta^1 \ \ \bfI^{12}+\frac{\bm\beta^1}{Z^{12}}+\frac{\bm\beta^1Y^{12}}{2} \ \ \bfI^{13}+\frac{\bm\beta^1-\bm\beta^2}{Z^{13}}+\frac{\bm\beta^1Y^{13}}{2}\nonumber \\   & \quad \quad \V^2 \ \ \bfI^{21}-\bm\beta^1/Z^{12} \ \ \bfI^{23}-\bm\beta^2/Z^{23}].
\end{align}
The additive errors due to attacks are 
\begin{align} 
\bar{D}
=[ \bm\beta^1 \ \ \bm 0 \ \ \bm\beta^2]\bar{W}^T= [\|\bar{W}_1\|\bm\beta^1 \ \ \bm 0 \ \ \|\bar{W}_3\|\bm\beta^2]W^T.
\end{align}

\subsection{Connections to existing work}\label{sec:app}

The detection of unobservable cyber data attacks has not been much addressed. \cite{SJ13} and \cite{LEDEH14} considered the detection of unobservable attacks to SCADA data and provided numerical results. \cite{SJ13}  assumes the measurements across time are i.i.d. distributed and detects the attacks based on statistical learning.   \cite{LEDEH14}  assumes the SCADA measurements under DC power flow model are low-rank and proposes to detect the attacks by decomposing  a low-rank matrix and a sparse matrix from their sum. Our work differs from \cite{LEDEH14} in that we assume the intruder constantly injects data attacks to the same set of PMUs, while \cite{LEDEH14} assumes the intruder attacks different PMUs at different time instants. Furthermore, we provide the theoretical guarantee of our detection method.   

Our problem formulation of matrix decomposition is closely related to those in \cite{XCS12} and  \cite{MMG13}. When $W$ is an identity matrix, our problem reduces to the one in \cite{XCS12}. The difference between our model and the one in \cite{MMG13} is that  the sparse matrix $\bar{C}$ in \cite{MMG13} has nonzero entries located independent of each other, while $\bar{C}$ here is a column-sparse matrix. Our method and analysis are built up those in  \cite{XCS12}, but we consider a more general framework of matrix decomposition through the introduction of the transform matrix $W$. 

The significance of our work is twofold. First, we for the first time consider the case that the additive error matrix $\bar{D}$ can be dense (i.e., $W$ is a dense matrix), while the error matrices in  \cite{XCS12} and  \cite{MMG13} are sparse. We   show through both theoretical analysis and numerical experiments that  it is possible to achieve matrix decomposition with dense $\bar{D}$. Second, when $\bar{D}$ is a column-sparse matrix itself (i.e., $W$ is sparse), our decomposition method outperforms those in \cite{XCS12} and  \cite{MMG13} (see Section \ref{sec:compare} and \ref{sec:PMU}) in the sense that our recovery method can tolerate  a higher level of corruption (i.e., large support size of $\bar{D}$). This advance results from exploiting (\ref{eqn:state}), which reduces the degree of freedom of $\bar{D}$.

Note that our method and analysis hold for an arbitrary   $W$ and can be applied to other domains that involve decomposing a matrix as in (\ref{eqn:problem}). As discussed in \cite{MMG13},   applications include unveiling network traffic anomalies \cite{LCD04,TJ03},  dynamic magnetic resonance imaging \cite{FNDRL06,GCSZ11}, face recognition \cite{BJ03}, and music analysis \cite{LGWWCM09,LW07}. 

\section{Attack Identification without  Noise}\label{sec:method}
\subsection{Identification method and guarantee}\label{sec:analysis}
\floatname{algorithm}{Method}
\begin{algorithm}
\begin{algorithmic}
\REQUIRE PMU measurements $M$ in $t$ instants; coefficient $\eta$; the set $\Omega$ of the locations $(i,j)$ of the observed entries.
\STATE Find   ($L^*$, $C^*$), the optimum solution to the following optimization problem   
\begin{equation} \label{eqn:opt1} 
{\hspace{-0.2in}\min\limits_{L \in \C^{t\times  p}, C\in \C^{t \times n}} \| L\|_{*} + \lambda \|C\|_{1,2} } 
\end{equation}
\begin{equation}\label{eqn:opt} 
\textrm{s.t.} \quad \sum_{i,j \in \Omega}|M_{ij}-L_{ij}-(CW^T)_{ij}|^2 \le \frac{|\Omega|}{tp}\eta^2
\end{equation}
\STATE Compute the SVD of $L^*=U^* \Sigma^* V^{* \dagger}$. 
\STATE Find column support of $D^*=C^*W^T$, denoted by  $\J^*$. 
\RETURN $L^*$, $C^*$, $L^*_{\J^{*c}}$, $U^*$ and $\J^*$.
\end{algorithmic}
\caption{Unobservable cyber attack identification method} 
\end{algorithm}

We first consider noiseless measurements ($\eta=0$). We assume a complete set of measurements for  analysis, but  our   method can be extended to cases when   measurements are partially lost. Moreover, although we consider attack patterns in (\ref{eqn:state}), our method can be generalized to detect combined attacks. In this case,  $\bar{D}$ is generalized to
\begin{equation}\label{eqn:Dg}
\bar{D}=\bar{C}W^T+\bar{S},
\end{equation}
where  a sparse matrix $\bar{S}$ represents attacks (observable and/or unobservable) that have different locations across time. Then (\ref{eqn:opt1})-(\ref{eqn:opt}) are generalized to 
\begin{equation} \label{eqn:opt3} 
{\hspace{-0.2in}\min\limits_{L \in \C^{t\times  p}, C\in \C^{t \times n}, S \in \C^{t\times  p}}  \| L\|_{*} + \lambda_1 \|C\|_{1,2} } + \lambda_2 \sum_{ij} |S_{ij} |
\end{equation}
\begin{equation}\label{eqn:opt4} 
\textrm{s.t.} \quad \sum_{i,j \in \Omega}|M_{ij}-L_{ij}-(CW^T)_{ij}-S_{ij}|^2 \le \frac{|\Omega|}{tp}\eta^2,
\end{equation}
with given  positive constants $\lambda_1$, $\lambda_2$. We study this extension numerically in Section \ref{sec:compare}.

To formally present the theoretical result, we need the following definitions.  
Given $\bar{L}=\bar{U}\bar{\Sigma}\bar{V}^\dagger$ and $W$, we define
\begin{equation} 
\epsilon:=\|\bar{V}^\dag W^{\ddagger}\|_{\infty,2}, \text{ } \mu:=\max_{i \neq j}\|W_i^{\dagger} W_j\|,
\end{equation}
\begin{equation} \label{eqn:sigmak}
\text{ and } \sigma_k:=\max_{\I: |\I| \leq k}\| (W^{\dagger}_\I W_\I)^{-1}\|.
\end{equation}
Note that $\sigma_1=1$ as $W$ has unit-norm columns, and $\epsilon$ depends on the rank $r$ of $\bar{L}$, since $\|\bar{V}\|_{F}^2=r$.

Pick any constants $\tilde{\psi}$ and $c$  in $(0,1)$ such that 
\begin{equation}\label{eqn:psistar}
(2-\tilde{\psi})\sqrt{\tilde{\psi}}/(1-\tilde{\psi}) \leq \sqrt{(1+c)/(1-c)}.
\end{equation}
For any integer $k$, define
\begin{equation}\label{eqn:lambdamin}
\lambda_{\min,k}=\frac{(1+(2-\tilde{\psi})^{-1})\epsilon}{1-(1+(2-\tilde{\psi})^{-1})k\sigma_k\mu}
\end{equation}
\begin{equation}\label{eqn:lambdamax}
\text{ and } \lambda_{\max,k}=\sqrt{\tilde{\psi}/(k\sigma_k)}.
\end{equation}
Our detection method is summarized in Method 1. (\ref{eqn:opt}) is a convex program and can be solved efficiently by generic solvers such as CVX\cite{CVX}. Its recovery guarantee is as follows. 
\begin{theorem}\label{thm:recover}
Suppose there exists nonzero $\tilde{k}$ such that 
\begin{equation}
\tilde{k} \mu \leq c,  \textrm{ and } \lambda_{\min,\tilde{k}} \leq \lambda_{\max,\tilde{k}},
\end{equation}
with $c$, $\lambda_{\min,\tilde{k}}$, and $\lambda_{\max,\tilde{k}}$ defined in (\ref{eqn:psistar})-(\ref{eqn:lambdamax}).
Then as long as the column support of $\bar{C}$ has size at most $\tilde{k}$, for any $\lambda \in [\lambda_{\min,\tilde{k}}, \lambda_{\max,\tilde{k}}]$, the output of Method 1 satisfies 
\begin{equation} \label{eqn:space}
U^*U^{*\dag}=\bar{U}  \bar{U}^\dag, 
\end{equation}
\begin{equation} 
\J^*=\bar{\J}\label{eqn:index} \text{ and } L^*_{\J^{*c}}=\bar{L}_{\bar{\J}^{c}}. \nonumber
\end{equation}
\end{theorem}

Theorem \ref{thm:recover} guarantees that the affected PMUs can be correctly located and thus,  the ``clean'' PMU measurements could be identified. Furthermore, the subspace spanned by the actual phasors can be recovered. Since we do not obtain any actual measurements from PMUs that are under attack, it is impossible to recover the exact measurements in the affected PMUs without further regularization. Under the conditions of Theorem  \ref{thm:recover},   the recovery is also successful when the column support of $\bar{C}$ is zero. Thus, the false alarm rate is zero.

Method 1 is  motivated by \cite{XCS12}. In fact, after post-multiplying $W^{\ddagger} (W^TW^{\ddagger})^{-1}$ to both sides of (\ref{eqn:M}), we have
\begin{equation} \nonumber
MW^{\ddagger} (W^TW^{\ddagger})^{-1}=\bar{L}W^{\ddagger} (W^TW^{\ddagger})^{-1} + \bar{C} + NW^{\ddagger} (W^TW^{\ddagger})^{-1}
\end{equation}
where the right-hand side is the sum of a low-rank matrix plus a column-sparse matrix and noise. 
Then, the results of \cite{XCS12} can be directly applied to our problem. We do not follow this path due to  two reasons. First, $MW^{\ddagger}(W^TW^{\ddagger})^{-1}$ cannot be computed if some entries of $M$ are missing, while Method 1 can be easily extended to scenarios with missing data  by restricting the constraints in (\ref{eqn:opt}) to the observed measurements. Second, $(W^TW^{\ddagger})^{-1}$ does not exist when $W$ is a flat matrix, i.e., $p < n$, while Method 1 and Theorem \ref{thm:recover} can be applied to an arbitrary $W$.

\subsection{Discussion of $\lambda$ and $\tilde{k}$}\label{sec:lambda}

We remark that due to the slackness in the proof, $\lambda \in [\lambda_{\min,\tilde{k}}, \lambda_{\max,\tilde{k}}]$ in Theorem \ref{thm:recover} is sufficient but not necessary\footnote{Specially, the requirements on dual certificate in Lemma \ref{lem:certificate} are sufficient but not necessary. Furthermore, we use loose bounds in the proofs to simplify analysis. $\epsilon$, $\mu$, and $\sigma_k$ are in turn defined based on worst-case scenarios.}. There may exist $\lambda$ outside $[\lambda_{\min,\tilde{k}}, \lambda_{\max,\tilde{k}}]$ that can still lead to correct recovery. We observe from numerical experiments that recovery performance is generally much better than the bound in  Theorem \ref{thm:recover}. Furthermore, when $L$ is fixed, as $\tilde{k}$ decreases, $\lambda_{\min,\tilde{k}}$ decreases, and $\lambda_{\max,\tilde{k}}$ increases. Thus, intuitively, if the number of affected PMUs decreases, a wider range of $\lambda$ is proper for Method 1. For a detailed  discussion, we state the following lemma and defer its proof to the Appendix. 
\begin{lemma}\label{lem:sigmak}
Suppose $k\mu < 1$, then $\sigma_k \leq (1-(k-1)\mu)^{-1}$.
\end{lemma}
Since $\sigma_k$ increases in $k$,   $\sigma_1 \geq 1$, and $k \mu \leq c<1$, % for constant $c$,
together  with Lemma \ref{lem:sigmak},  we know $\sigma_{\tilde{k}} = \Theta(1)$\footnote{We use the notations  $g(n)\in O(h(n))$, $g(n) \in \Omega(h(n))$, or $g(n)=\Theta(h(n))$ if as $n$ goes to infinity, $g(n) \leq c \cdot h(n)$, $g(n) \geq c  \cdot h(n)$ or $c_1 \cdot h(n) \leq g(n) \leq c_2 \cdot h(n)$ eventually holds for some positive constants $c$, $c_1$ and $c_2$ respectively.}. Since $\tilde{\psi}$ is a constant, one can check that $\lambda_{\min,\tilde{k}}= \Theta(\epsilon)$, and $\lambda_{\max,\tilde{k}}= \Theta(\sqrt{1/\tilde{k}})$. 

Note that $\|\bar{V}^ {\dagger}\|_F^2 =r$. We assume that $\|\bar{V}^ {\dagger}\|$ is column-incoherent \cite{XCS12} with some positive constant $\rho>1$, i.e., 
\begin{equation}\label{eqn:rho}
\|\bar{V}^ {\dagger}\|_{\infty, 2} \le \sqrt{\rho r/p}.
\end{equation}
We assume the number of PMU channels incident to each bus is in the range of $[d, Cd]$ for some $d>0$ and some constant $C$. This is also the number of nonzero entries in each column of $W$ with unit column-norm. Then $p=\Theta(dn)$, and we have
\begin{align}
\epsilon =  \|\bar{V}^\dag W^{\ddagger}\|_{\infty,2}  
\leq &   \sqrt{\frac{\rho r}{p}} \max_{i} \sum_{j} |W_{ij}|  = O(\sqrt{\frac{r}{n}}). \label{eqn:ep}
\end{align} 
Therefore, as long as $\tilde{k} =O(n/r)$, when $n$ is sufficiently large, $\lambda_{\min,\tilde{k}}  \leq \lambda_{\max,\tilde{k}}$. $\tilde{k} \mu \leq c$ requires that $\tilde{k} =O(1/\mu)$. Note that $\mu= \Theta(\frac{1}{d})$. Thus, if both $\tilde{k} =O(n/r)$ and $\tilde{k}=O(d)$ hold, then a proper $\lambda$ exists, and Theorem \ref{thm:recover} holds. 
 
In the  case that $d= \Theta(n)$, $\tilde{k}$ could be  $\Theta(n/r)$. If $r$ is a constant, our method succeeds even when  a constant fraction of bus voltages are corrupted. Also consider the   case that $\tilde{k}=1$. We pick $\tilde{\psi}$ and $c$ in (\ref{eqn:psistar}) arbitrarily close to one, then $\lambda=1$ is a proper choice (see Fig. \ref{fig:case2} for results on actual PMU data) provided that $\epsilon +\mu \leq 0.5$. Since $\epsilon$ scales as $1/\sqrt{n}$ and $\mu$ scales as $1/d$,  the condition will be met in large systems that are tightly connected. 
Intuitively, $\mu$ is small if the bus degree is high, and the line impedances are in the same range. 

\begin{figure}[h] 
\begin{center}
\includegraphics[height=0.4\linewidth]{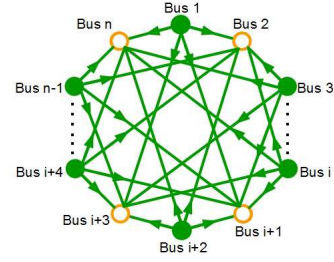}
\caption{$n$-bus ring network} \label{fig:mesh}
\end{center}
\end{figure}
We next use an example to illustrate the existence of proper $\lambda$. Consider an $n$-bus ($n$ is even) ring network  in Fig.~\ref{fig:mesh}.  Each odd-numbered bus is connected to all even-numbered buses. There is no connection among odd buses and no connection among even-numbered buses. A PMU is installed on each odd bus and measures the corresponding voltage phasor and all incident line current phasors. For the simplicity of analysis, we assume $Z^{ij}=1$ and $Y^{ij}=0$ in this ring network. $W$ is a $(\frac{n^2}{4}+\frac{n}{2}) \times n$ matrix with unit norm columns. Specifically, for every integer $k$,  

$W_{ij} =
\begin{cases} 
\sqrt{2/(n+2)},  & \mbox{if }i\in \mathcal{I}_{k1}\mbox{ and }j=2k-1 \\
-\sqrt{2/n}, & \mbox{if }i\in \mathcal{I}_{k2}\mbox{ and }j=2k \\
0, & \mbox{otherwise}
\end{cases}$,

\noindent where 
\begin{equation*}
\mathcal{I}_{k1}:=\left\{ k+(k-1)\frac{n}{2}+k' ~|~ \textrm{interger } k'=0,1,2,...,\frac{n}{2} \right\},
\end{equation*}
\begin{equation*}
\mathcal{I}_{k2}:=\left\{ k+1+(\frac{n}{2}+1)k' ~|~ \textrm{interger }k'=0,1,2,...,\frac{n}{2}-1 \right\}.
\end{equation*} 
Note that $|\I_{k1}|=\frac{n}{2}+1$, $|\I_{k2}|=\frac{n}{2}$ for all $k$. Here, $\mu=2/\sqrt{n^2+2n}$. Then we have
\begin{equation}\label{eqn:column}
(V^{\dagger}W^{\ddagger})_j =
\begin{cases} 
\sqrt{2/(n+2)}\sum_{i\in \mathcal{I}_{k1}}(V^{\dagger})_i,  & \mbox{if }j=2k-1 \\
-\sqrt{2/n}\sum_{i\in \mathcal{I}_{k2}}(V^{\dagger})_i, & \mbox{if }j=2k 
\end{cases},
\end{equation}
where $V \in \C^{(\frac{n^2}{4}+\frac{n}{2}) \times r}$ contains the right singular vectors of the rank-$r$ measurement matrix $\bar{L} \in \C^{t \times (\frac{n^2}{4}+\frac{n}{2})}$.  If $\|\bar{V}^ {\dagger}\|$ is column-incoherent \cite{XCS12} with some positive constant $\rho$, then 
\begin{align}
\epsilon =  \|\bar{V}^\dag W^{\ddagger}\|_{\infty,2} \leq & \max \big(\sqrt{\frac{2}{n+2}} |\I_{k1}|, \sqrt{\frac{2}{n}} |\I_{k2}|\big) \cdot\|\bar{V}^ {\dagger}\|_{\infty, 2}  \nonumber \\
\leq & \sqrt{\frac{n+2}{2} }\cdot \sqrt{\frac{\rho r}{(\frac{n}{2}+1)\frac{n}{2}}} \leq \sqrt{\frac{2\rho r}{n}}, \label{eqn:ep1}
\end{align} 
where the first inequality follows from  (\ref{eqn:column}), and the second inequality follows from (\ref{eqn:rho}). 

To find $\lambda$, we pick $c=1/4$ and $\tilde{\psi}=1/8$. We choose $\tilde{k} = \frac{n}{48\rho r}$. One can check that (\ref{eqn:psistar}) follows. Then
\begin{equation}
\tilde{k} \mu = \frac{n}{48\rho r} \times \frac{2}{\sqrt{n^2+2n}} \leq \frac{1}{24 \rho r} \leq \frac{1}{24} \leq \frac{1}{4}=c,
\end{equation}
where the last inequality follows since $\rho>1$ and $r \geq 1$.
Then from Lemma \ref{lem:sigmak}, we have 
\begin{equation}
\sigma_{\tilde{k}} \leq (1-(\tilde{k}-1)\mu)^{-1} \leq (1-\tilde{k} \mu)^{-1} \leq 24/23.
\end{equation}
From  (\ref{eqn:lambdamin}) and (\ref{eqn:lambdamax}),
\begin{equation}
\lambda_{\min,\tilde{k}}\leq \frac{(1+(2-\tilde{\psi})^{-1})\epsilon}{1-(1+(2-\tilde{\psi})^{-1})\tilde{k}\mu \sigma_{\tilde{k}}}
\leq \frac{23\epsilon}{14} \leq \frac{23}{14}\sqrt{\frac{2\rho r}{n}} .
\end{equation} 
%where we use $k^* \mu \leq 1/24$ in the derivation. From (\ref{eqn:lambdamax}),
\begin{equation}
\lambda_{\max,\tilde{k}}= \sqrt{\frac{1/8}{\frac{n}{48\rho r} \sigma_{\tilde{k}}}} \geq \frac{1}{2} \sqrt{\frac{23\rho r}{n}}.
\end{equation} 
Since $\frac{23}{14}\sqrt{\frac{2\rho r}{n}} <\frac{1}{2}\sqrt{\frac{23\rho r}{n}}$, then $\lambda_{\min,\tilde{k}}<\lambda_{\max,\tilde{k}}$. Then there exists $\lambda$ such that Method 1 correctly identifies the corruptions in up to $\tilde{k} = \frac{n}{48\rho r}$ bus voltages. In fact, any $\lambda \in [\frac{23}{14}\sqrt{\frac{2\rho r}{n}},\frac{1}{2}\sqrt{\frac{23\rho r}{n}}]$ suffices. Note that for a constant $r$, $\tilde{k}$ is linear in $n$, the total number of buses.  

\subsection{Proof sketch of Theorem \ref{thm:recover}}
The proof of Theorem \ref{thm:recover} follows the same line as the proof of Theorem 1 in \cite{XCS12}. With the additional projection matrix $W$, our proof is more involved than the one in \cite{XCS12}. 
 
Like \cite{XCS12}, we design the following Oracle Problem (\ref{eqn:oracle}) by adding explicit constraints that the solution pair should have the correct column space of $\bar{L}$ and the correct column support of $\bar{C}$. The major step is to show that an optimal solution ($L^*$,$C^*$) to (\ref{eqn:opt}) is also an solution to the Oracle problem (\ref{eqn:oracle}). Note that Oracle problem is only designed for analysis, since $\bar{U}$ and $\bar{\I}$ are unknown to the operator.  
\begin{equation}\label{eqn:oracle}
\begin{tabular}{p{0.9in}p{0.1 in}p{1.6 in}} 
 Oracle Problem & $\min\limits_{L, C}$ %\in \C^{t\times  p},C\in \C^{t \times n}}$ 
 &  $\| L\|_{*} + \lambda \|C\|_{1,2}$ \\ 
    & s.t. & $M = L+ C W^T $\\
        && $\P_{\bar{U}}(L)=L$, $\P_{\bar{\I}}(C)=C$. 
\end{tabular}
\end{equation}

Let $(L', C')$ be an optimal solution to the Oracle problem (\ref{eqn:oracle}). We define $\P_{T'} := \P_{U'} +\P_{V'}-\P_{U'}\P_{V'}$, where the SVD of $L'= U' \Sigma' V'^{\dag}$. Define
\begin{align}\nonumber
\G(C'):=\{  H \in \C^{t \times k} ~|~  \forall i\in \I' : H_i= C_i'/\|C_i'\| ; \\  \forall i \in \bar{\I} \cap (\I')^c: \|H_i\|_2 \leq 1\},\nonumber
\end{align}
where $\I'$ is the column support of $C'$. We have
\begin{lemma}[Lemma 4 and Lemma 5 of \cite{XCS12}]
\begin{equation} \nonumber 
U'U'^\dag=\bar{U}\bar{U}^\dag.
\end{equation} 
There exists an orthonormal matrix $\hat{V} \in \C ^{t \times p}$ such that 
\begin{equation}\label{eqn:UU'}
U' V'^\dag =\bar{U} \hat{V}^\dag. 
\end{equation} 
Also, we have
\begin{equation}\nonumber
\P_{T'}:= \P_{U'}+\P_{V'}-\P_{U'}\P_{V'}=\P_{\bar{U}}+\P_{\hat{V}}-\P_{\bar{U}}\P_{\hat{V}}.
\end{equation}
\end{lemma}

The following lemma establishes that the solution to  the Oracle problem (\ref{eqn:oracle}) is also a solution to (\ref{eqn:opt}),
\begin{lemma}\label{lem:cond}
An optimal solution $(L', C')$ to (\ref{eqn:oracle}) is an optimal solution to (\ref{eqn:opt}) if there exists $Q \in \C^{t\times p}$ that satisfies
\begin{equation}\label{eqn:cond}
\begin{aligned}
&(a) \P_{T'} (Q) = U' V'^{\dagger}, 
\quad\quad (b)  \| \P_{T'^{\perp}} (Q)\| \leq 1 ,\\
&(c)   (QW^{\ddagger})_{\bar{\I}}/\lambda \in \G(C'),  
\quad \textrm{and }(d) \|(QW^{\ddagger})_{\bar{\I}^c} \|_{\infty,2} \leq \lambda.
\end{aligned}
\end{equation}
If both (b) and (d) are strict, and $\P_{\bar{\J}} \cap \P_{V'}=\{0\}$, then any optimal solution $(L^*,C^*)$ to (\ref{eqn:opt}) satisfies $\P_{\bar{U}}(L^*)=L^*$, $\P_{\bar{\I}}(C^*)=C^*$. 
\end{lemma}

The major technical step is to construct $Q$, called the \textit{dual certificate},  that satisfies (\ref{eqn:cond}).  Our construction method is as follows. Pick $\hat{H} \in \G(C')$ that satisfies 
\begin{equation}\label{eqn:hath} 
\hat{V}^\dag W^{\ddagger}_{\bar{\I}} = \lambda\bar{U}^\dag \hat{H}.
\end{equation}
Define
\begin{equation}\label{eqn:phi}
\Phi := \lambda \hat{H} (W_{\bar{\I}}^{T}W^{\ddagger}_{\bar{\I}})^{-1}W_{\bar{\I}}^{T}, \text{ } \Delta_1:=\P_{\bar{U}}(\Phi),
\end{equation}
\begin{equation} \label{eqn:delta2}
\Delta_2:=\P_{\bar{U}^\perp}(I-\P_{W_{\bar{\I}}})\P_{\hat{V}}(I+\sum_{i=1}^\infty (\P_{\hat{V}}\P_{W_{\bar{\I}}}\P_{\hat{V}})^i)\P_{\hat{V}}(\Phi), 
\end{equation}
\begin{equation} \label{eqn:PWI}
\text{ where } \P_{W_{\bar{\I}}}(X):= X W^{\ddagger}_{\bar{\I}}(W_{\bar{\I}}^{T}W^{\ddagger}_{\bar{\I}})^{-1}W_{\bar{\I}}^{T}.
\end{equation}
\begin{equation}\label{eqn:Q}
  \quad \quad Q:= \bar{U} \hat{V}^\dag +\Phi-\Delta_1-\Delta_2.
\end{equation}
We show that $Q$ in (\ref{eqn:Q}) is well defined in Appendix-B. Lemma \ref{lem:certificate} shows  that $Q$  in (\ref{eqn:Q}) is the desired dual certificate. 
\begin{lemma}\label{lem:certificate}
Suppose there exists nonzero $\tilde{k}$ such that $\tilde{k} \mu \leq c$ for $c$ in (\ref{eqn:psistar}), and $\lambda_{\min,\tilde{k}} \leq \lambda_{\max,\tilde{k}}$ with $\lambda_{\min,\tilde{k}}$ and $\lambda_{\max,\tilde{k}}$ defined in (\ref{eqn:lambdamin}) and (\ref{eqn:lambdamax}). Then as long as the column support of $\bar{C}$ has size at most $\tilde{k}$, for any $\lambda \in [\lambda_{\min,\tilde{k}}, \lambda_{\max,\tilde{k}}]$, $Q$ defined in (\ref{eqn:Q}) satisfies (\ref{eqn:cond}).
\end{lemma}

Theorem \ref{thm:recover} follows when we combine Lemmas \ref{lem:cond} and \ref{lem:certificate}. Please refer to the Appendix for the proofs. 
\section{Attack Identification with Noise}\label{sec:method1}

We now analyze the detection performance when $M$ contains noise ($N \neq 0$) with $\|N\|_F \le \eta$. Given $k$, define
\begin{equation}\nonumber
\lambda'_{\min,k}=\frac{(1+(2-\tilde{\psi})^{-1})\epsilon}{1/2-(1+(2-\tilde{\psi})^{-1})k\sigma\mu}, \text{ and } \lambda'_{\max,k}= \frac{1}{2}\sqrt{\frac{\tilde{\psi}}{k\sigma_k}}. 
\end{equation}
\begin{theorem}\label{thm:recover1}
Suppose there exists nonzero $\tilde{k}$ such that $\tilde{k} \mu \leq c$ for $c$ in (\ref{eqn:psistar}), and $\lambda'_{\min,\tilde{k}} \leq \lambda'_{\max,\tilde{k}}$. 
Then if column support size of $\bar{C}$ is at most $\tilde{k}$, for any $\lambda \in [\lambda'_{\min,\tilde{k}}, \lambda'_{\max,\tilde{k}}]$, there exists a pair $(\tilde{L}$, $\tilde{C})$, where $\tilde{L}+\tilde{C}W^T = \bar{L}+\bar{C}W^T$, $\P_{\bar{U}}(\tilde{L})=\tilde{L}$ and $\P_{\bar{\I}}(\tilde{C})=\tilde{C}$, such that the output of Method 1 satisfies
\begin{align}\label{eqn:L}\nonumber
& \|L^*-\tilde{L}\|_F \\ 
\le & (2-\tilde{\psi}+\frac{\lambda+(2-\tilde{\psi})\sqrt{1+(n-1)\mu}}{\lambda}\sqrt{\theta+3r})\frac{2\eta}{1-\tilde{\psi}},
\end{align}
\vspace{-5mm}
\begin{align}\label{eqn:C}\nonumber
&\text{and }\|C^*-\tilde{C}\|_F \\ \nonumber
\le & (1+(\frac{\lambda+\sqrt{1+(n-1)\mu}}{\lambda}+\frac{1-\tilde{\psi}}{\lambda\sigma_k\sqrt{1+(k-1)\mu}})\sqrt{\theta+3r})\\
& \frac{2\eta\sigma_k\sqrt{1+(k-1)\mu}}{1-\tilde{\psi}},
\end{align}
where $\theta := \min(t, p)$.
\end{theorem}
The discussion of the existence of $\lambda$ is very similar to the discussion for Theorem \ref{thm:recover}, so we skip it. If $\tilde{k} \mu \leq c$ and $\tilde{k}=O(n/r)$ hold, then a proper $\lambda$ exists. Theorem \ref{thm:recover1} guarantees that ($L^*$, $C^*$) returned by Method 1 is ``close'' to a pair that has the correct column space and column support, and the distance measured by Frobenius norm is proportional to the noise level $\eta$. The proof of Theorem \ref{thm:recover1} follows the same line as the proof of Theorem 2 in \cite{XCS12} mostly with modifications to address the projection matrix $W$. We establish Lemma \ref{lem:cond1}, a counterpart in the noisy case of Lemma \ref{lem:cond}, that demonstrates that Method 1 succeeds if there exists a dual certificate $Q$ with tighter requirements than that in the noiseless case.

\begin{lemma}\label{lem:cond1}
There exists $(\tilde{L}$, $\tilde{C})$ where $\tilde{L}+\tilde{C}W^T = \bar{L}+\bar{C}W^T$, $\P_{\bar{U}}(\tilde{L})=\tilde{L}$, $\P_{\bar{\I}}(\tilde{C})=\tilde{C}$, such that the output of Method 1 satisfies (\ref{eqn:L}) and (\ref{eqn:C}),
if there exists $Q \in \C^{t\times p}$ that satisfies
\begin{equation}\label{eqn:cond1}
\begin{aligned}
&(a) \P_{\bar{T}} (Q) = \bar{U} \bar{V}^{\dagger}, 
\quad\quad (b)  \| \P_{\bar{T}^{\perp}} (Q)\| \leq 1/2 ,\\
&(c)   (QW^{\ddagger})_{\bar{\I}}/\lambda \in \G(\bar{C}),  
\quad \textrm{and }(d) \|(QW^{\ddagger})_{\bar{\I}^c} \|_{\infty,2} \leq \lambda/2.
\end{aligned}
\end{equation}
\end{lemma}

The construction of $Q$ is the same as that in Section \ref{sec:method} (equations (\ref{eqn:hath}) to (\ref{eqn:Q})). We show that $Q$ is the desire dual certificate if $\lambda$ belongs to $[\lambda'_{\min}, \lambda'_{\max}]$ in Lemma \ref{lem:certificate1}.

\begin{lemma}\label{lem:certificate1}
If the column support size of $\bar{C}$ is at most $\tilde{k}$, then for any $\lambda \in [\lambda'_{\min}, \lambda'_{\max}]$, $Q$ defined in (\ref{eqn:Q}) satisfies (\ref{eqn:cond1}).
\end{lemma}

Theorem \ref{thm:recover1} follows when we combine Lemmas \ref{lem:cond1} and \ref{lem:certificate1}. Please refer to the Appendix for the proofs.
\section{Simulation}\label{sec:simu}

We explore the performance of data attack identification methods on both synthetic data and actual PMU data from the Central NY power system. We use CVX \cite{CVX} to solve (\ref{eqn:opt}). We identify a column of $C^*$ to be nonzero if its $\ell_2$ norm exceeds the predefined threshold $\epsilon_1$. Method 1 succeeds if $\|U^*U^{*\dagger}-\bar{U}\bar{U}^{\dagger}\|\leq \epsilon_2$ for some small positive $\epsilon_2$, and the column supports of $\bar{C}$ and $C^*$ are the same. 

\subsection{Performance on synthetic data}

Fix $t=p=50$. Given rank $r$, we generate matrices $A$ $\in$ $\R^{t \times r}$ and $B$ $\in$ $\R^{p \times r}$ with each entry  independently drawn from Gaussian $\mathcal{N}(0,1)$ and set $\bar{L} := A   B^{T}$. We generate matrix $W$ $\in$ $\R^{p \times n}$ with independent $\mathcal{N}(0,1)$ entries. To generate a column-sparse matrix $\bar{C} \in \R^{t \times n}$, we randomly select the column support and set the nonzero entries to be independent $\mathcal{N}(0,1)$.  We vary $r$ and the number of corrupted columns, and take 100 runs for each case. $\lambda$ is set to be 0.95.

\subsubsection{Noiseless formulation}
We simulate the observed measurement matrix $M$ according to (\ref{eqn:problem}) with $N=0$. We apply Method 1 to obtain the estimation $(L^*,C^*)$. We set $\epsilon_1$ and $\epsilon_2$ to be 0.002 and 0.01, respectively. Fig.~\ref{fig:gray} shows the transition property of Method 1 in gray scale. White stands for 100\% success while black denotes 100\% failure. When $n$ is 25, $W$ is a tall matrix ($p>n$). When $n$ is 100, $W$ is a flat matrix ($p<n$). For both simulations, the identification is successful even when rank $r$ is six, and $\bar{C}$ has two nonzero columns. 

\begin{figure}[h] 
\begin{center}
\includegraphics[height=0.43\linewidth]{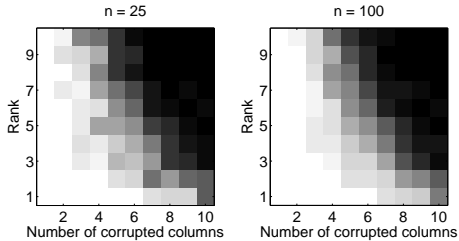}
\caption{Matrix decomposition performance for different $n$} \label{fig:gray} \vspace{-3 mm} 
\end{center}
\end{figure}

We further assume some of the observations are missing. We generate $M$ as before and then delete some randomly selected entries. Fig. \ref{fig:graynew} shows the decomposition performance of Method 1 for partial observation. We can see that the successful  decomposition rate is close to the complete observation case even only 80\% of the entries are observed. 

\begin{figure}[h] 
\begin{center}
\includegraphics[height=0.43\linewidth]{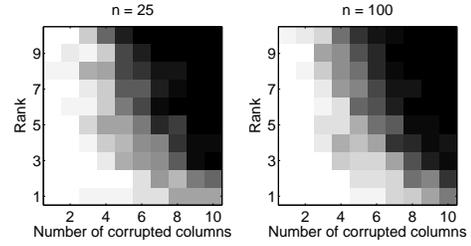}
\caption{Matrix decomposition performance for different $n$ with 80\% observed entries} \label{fig:graynew} \vspace{-7mm} 
\end{center}
\end{figure}  

\subsubsection{Noisy formulation}

We generate matrix $N \in \R^{t \times p}$ with independent Gaussian $\mathcal{N}(0,\sigma^2)$ entries. We fix the matrix rank $r$ to be 3 and the number of corrupted columns to be 3. We simulate the observed measurement matrix $M$ according to (\ref{eqn:problem}). We set $\eta$ to be $\|N\|_F$ and apply Method 1 to obtain the estimation $(L^*,C^*)$. $\epsilon_1$ is set to be 0.001.

\begin{figure}[h] 
\begin{center}
\psfrag{Noise level sigma}[][][0.65]{Noise level $\sigma$}
\psfrag{noise level sigma}[][][0.65]{Noise level $\sigma$}
\psfrag{Succeed rate}[][][0.65]{Succeed rate}
\psfrag{Difference of the column space}[][][0.48]{Difference of the column space}
\psfrag{(a) Column space}[][][0.65]{(a) Column space}
\psfrag{(b) Set of corrupted columns}[][][0.65]{(b) Set of corrupted columns}
\psfrag{difference}[][][0.6]{$\|U^*U^{*\dagger}-\bar{U}\bar{U}^{\dagger}\|$}
\includegraphics[height=0.3\linewidth]{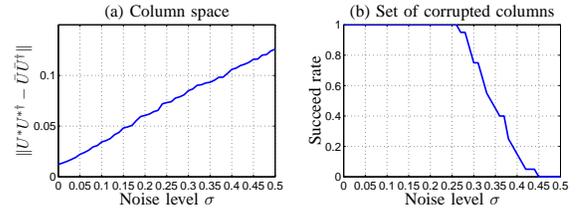}
\caption{Performance of Method 1 for different noise level $\sigma$} \label{fig:udiff} \vspace{-2mm}
\end{center}
\end{figure}

Fig. \ref{fig:udiff} shows the difference between the original and reconstructed column space ($\|U^*U^{*\dagger}-\bar{U}\bar{U}^{\dagger}\|$) and the succeed rate for determining the set of corrupted columns according to different noise level $\sigma$. We can see that Method 1 can successfully identify the corrupted columns when the noise level $\sigma$ is below 0.25. Method 1 can recover the column space with small errors when $\sigma$ is smaller than 0.1. 

\subsection{Comparison with other methods on synthetic data}\label{sec:compare}

\subsubsection{{\small $\bar{D}=\bar{C}W^T$} is column-sparse}

Refs. \cite{XCS12} \cite{LEDEH14} considered matrix decomposition problem when $\bar{D}$ is column-sparse and scattered-sparse, respectively. We compare our method with them in the special case that $\bar{D}=\bar{C}W^T$ is column-sparse. Fix $t=p=50$, $n=20$, and $r=2$. We generate $\bar{L}$ and $\bar{C}$ with the same rules as in Section V-A. We generate a binary matrix $W$ $\in$ $\R^{p \times n}$ with two `1's each row  and five `1's each column. Then the ratio of support sizes of $\bar{D}$ and $\bar{C}$ is about five. $\bar{D}$ is column-sparse when $\bar{C}$ is column-sparse. We simulate the measurement matrix $M$ according to (\ref{eqn:problem}) with $N=0$. 
$\lambda$ in our method is set to be 0.9. $\lambda$'s in methods of \cite{XCS12} and \cite{LEDEH14} are set to be 0.5 and 0.1, respectively.

\begin{figure}[h] 
\begin{center}
\psfrag{Number of outliers}[][][0.7]{Number of corrupted columns in $\bar{C}$}
\psfrag{success rate}[][][0.7]{Success rate}
\includegraphics[height=0.35\linewidth]{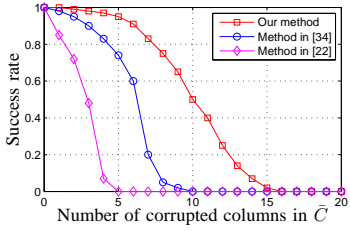}
\caption{Success rates when $\bar{D}=\bar{C}W^T$ is column-sparse.} \label{fig:compare} \vspace{-5mm}
\end{center}
\end{figure}

Fig.~\ref{fig:compare} shows the success rates of three methods  with different support sizes of $\bar{C}$. Our method  performs the best since we exploit the structure $\bar{D}=\bar{C}W^T$ besides sparsity. The false alarm rate of our method is zero.

\subsubsection{Combination of attack patterns.} We consider the general case that the attacks satisfy (\ref{eqn:Dg}). We use the generalized version   in (\ref{eqn:opt3})-(\ref{eqn:opt4}) to detect combined attacks. $\lambda_1$ and $\lambda_2$  in  (\ref{eqn:opt3}) are set to be 1 and 0.1, respectively. $\lambda$'s in methods of \cite{XCS12} and \cite{LEDEH14} are set to be 0.5 and 0.1, respectively. $\bar{L}$, $\bar{C}$, and $W$ are generated the same as above. $\bar{S}$ is a sparse matrix with nonzero entries independently drawn from $\mathcal{N}(0,1)$. We define the correct estimation of the column space of $\bar{L}$ as a successful recovery. Fig.~\ref{fig:compare2} compares the methods when $\bar{C}$ is a zero matrix. The attacks are scattered-sparse, and our method performs as well as that in \cite{LEDEH14}.  Fig.~\ref{fig:compare3} compares the methods when both column-sparse and scattered-sparse attacks exist. Besides  a sparse $\bar{S}$, we randomly select two columns in $\bar{C}$ and select their  entries independently from $\mathcal{N}(0,1)$. Only our method  succeeds when both attacks exist. 

\begin{figure}[h] 
\begin{center}
\psfrag{Average corruption rate}[][][0.7]{Percentage of corrupted entries in $\bar{S}$}
\psfrag{success rate}[][][0.7]{Success rate}
\includegraphics[height=0.35\linewidth]{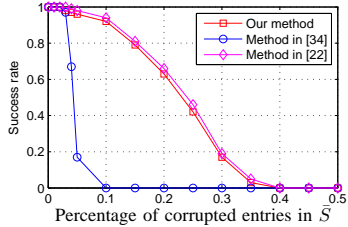}
\caption{Success rates when $\bar{D}=\bar{S}$ is scattered-sparse.} \label{fig:compare2}\vspace{-7mm}
\end{center}
\end{figure}

\begin{figure}[h] 
\begin{center}
\psfrag{Average corruption rate}[][][0.7]{Percentage of corrupted entries in $\bar{S}$ when $\bar{C}$ has two nonzero columns}
\psfrag{success rate}[][][0.7]{Success rate}
\includegraphics[height=0.35\linewidth]{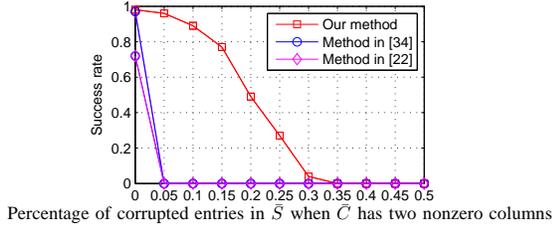}
\caption{Success rates when $\bar{D}=\bar{C}W^T+\bar{S}$.} \label{fig:compare3}\vspace{-7mm}
\end{center}
\end{figure}

\subsection{Performance comparison on actual PMU data}\label{sec:PMU}

We consider the PMU data shown in Section \ref{sec:pmu}. Two two-second PMU datasets are tested. One contains  ambient data, and the other contains an abnormal event ($t = 17-19s$ and $t = 2-4s$ in Fig. \ref{fig:current}, respectively). We first inject data attacks  as an intruder and then use Method 1 to detect the attacks.

We  consider the scenario that an intruder alters the PMU channels that measure $I^{12}$,$I^{52}$,$I^{13}$ and $I^{43}$ in order to corrupt voltage estimations of Buses 2 and 3. Fig.~\ref{fig:attack} visualizes the actual PMU data and the data after the injection of attacks for two 2-second datasets. $\eta$ and $\lambda$ are set to be 5 and 1 respectively in Method 1. Fig.~\ref{fig:PMUfig2} shows the $\ell_2$ norm of each column of the resulting $\bar{D}$ matrix. The columns with significant $\ell_2$ norm correspond to channels that measure $I^{12}$,$I^{52}$,$I^{13}$ and $I^{43}$. Therefore, our method successfully identifies the four PMU channels under attack.  We repeat the same experiment when an intruder alters the channels that measure  $V^{5}$, $I^{52}$, $I^{54}$, $I^{59}$, and  $I^{45}$ to corrupt voltage estimation of Buses 5. Fig.~\ref{fig:PMUfig3} shows the $\ell_2$ norm of each column of the resulting $\bar{C}$ matrix in this case. The column with significant $\ell_2$ norm corresponds Bus 5. Thus the recovery is also successful.
  
\begin{figure}[h] 
\begin{center}
\includegraphics[height=0.33\linewidth]{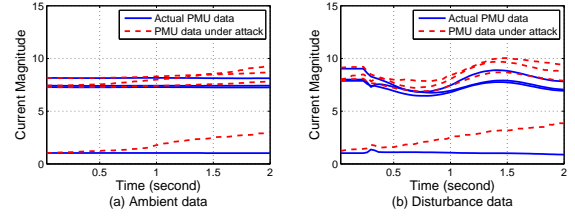}
\caption{The actual PMU data and PMU data under attack} \label{fig:attack} \vspace{-5mm}
\end{center}
\end{figure}

\begin{figure}[h] 
\begin{center}
\begin{psfrags}
\psfrag{l2 norm of the column}[][][0.6]{$\ell_2$ norm of the column}
\includegraphics[height=0.33\linewidth]{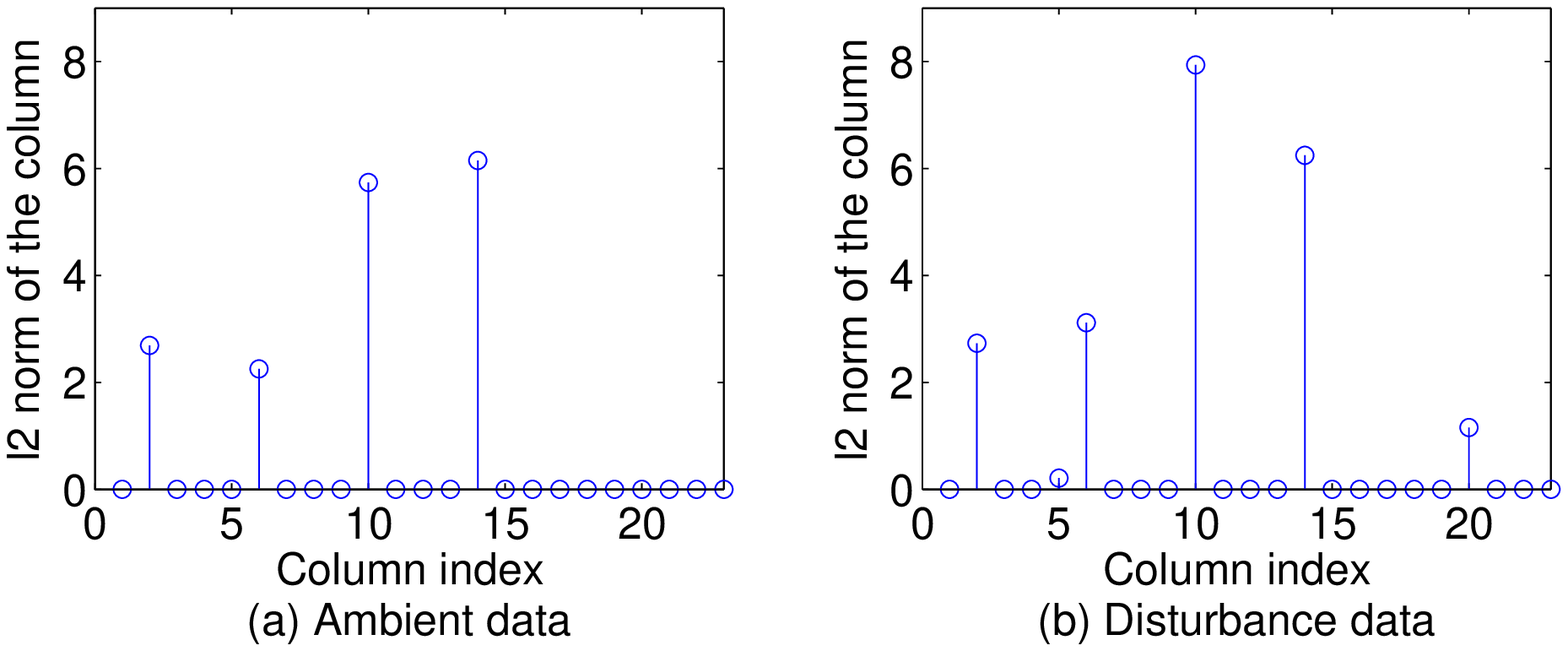}
\caption{$\ell_2$ norm of each column of $\bar{D}$} \label{fig:PMUfig2} \vspace{-5 mm} 
\end{psfrags}
\end{center}
\end{figure}

\begin{figure}[h] 
\begin{center}
\psfrag{l2 norm of the column}[][][0.6]{$\ell_2$ norm of the column}
\includegraphics[height=0.33\linewidth]{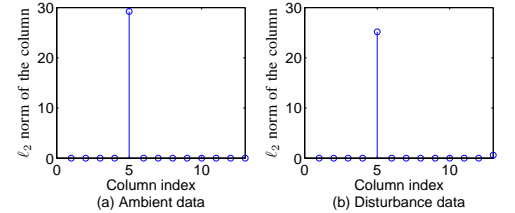}
\caption{$\ell_2$ norm of each column of $\bar{C}$} \label{fig:PMUfig3} \vspace{-5 mm} 
\end{center}
\end{figure}

{Fig. \ref{fig:case2} compares our method and that in  \cite{XCS12} on the ambient PMU data. Given support size of $\bar{C}$, the result is averaged over all possible attack locations. Our method outperforms \cite{XCS12} because we exploit (\ref{eqn:state}) to reduce the degree of freedom in $\bar{D}$. For example, 7 out of 23 channels needs to be attacked to change the state of Bus 1. That means 30\% of the columns of $\bar{D}$ are nonzero. This high percentage of corruption in $\bar{D}$ cannot be handedly by \cite{XCS12}.

\begin{figure}[h] 
\begin{center}
\begin{psfrags}
\psfrag{Number of buses under attack}[][][0.7]{Number of buses under attack}
\psfrag{Success rate}[][][0.7]{Success rate}
\includegraphics[height=0.33\linewidth]{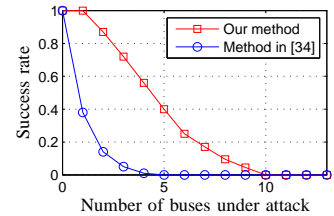}
\caption{Success rates with varying support size of $\bar{C}$, or equivalently, the number of affected system states. } \label{fig:case2}
\end{psfrags}
\end{center}
\end{figure}

\section{Conclusion}\label{sec:con}

We address the problem of detecting successive unobservable cyber data attacks to PMU measurements. We formulate the identification problem as a matrix decomposition problem of a low-rank matrix and a transformed column-sparse matrix. We propose a convex-optimization-based method and provide its theoretical guarantee. Although motivated by power system monitoring, our results on matrix decomposition can be applied to other scenarios. One future direction is the analysis of the detection performance when some of the measurements are lost during the communication to the central operator.

\section*{Acknowledgement}
We thank New York Power Authority for providing PMU data for the Central NY Power System. This research is supported in part by the ERC Program of NSF and DoE under the supplement to NSF Award   EEC-1041877 and the CURENT Industry Partnership Program, and in part by NSF Grant 1508875, NYSERDA Grants  \#36653 and \#28815. 
%\bibliographystyle{IEEEtranS}\vspace{-3 mm}
%\bibliography{./bibfiles/IEEEabrv,./bibfiles/ref,./bibfiles/MengWangPub}

\section*{Appendix}

\subsection{Proof of Lemma \ref{lem:sigmak}}
\begin{IEEEproof}
We first state the following result that will be used in the proof.
\begin{lemma}[Ger\v{s}hgorin circle theorem \cite{HJ12}]
Let A be a complex $n \times n$ matrix, with entries $a_{ij}$. Then, every eigenvalue of A lies within at least one of the Ger\v{s}hgorin discs  $D_i(A) (i=1,...,n)$, where
$D_i(A):=\{ z \in \C: |z-a_{ii}| \leq  \sum_{j\neq{i}} \left|a_{ij}\right|\}$.
\end{lemma}
For any given $\I$ with $|\I| \leq k$, since $W$ has unit-norm columns, and $|W^\dag_iW_j|\leq \mu$ for all $i \neq j$, from Ger\v{s}hgorin circle theorem, we have $\|I- W^\dag_{\I}W_{\I}\|\leq (k-1)\mu <1$, where the last inequality follows from $k\mu<1$. Then,
\begin{align*}
\|(W^\dag_{\I}W_{\I})^{-1}\| = &\|\sum_{i=0}^\infty (I- W^\dag_{\I}W_{\I})^i\| \leq \sum_{i=0}^\infty \|(I- W^\dag_{\I}W_{\I})^i\|\nonumber  \\
\leq & 1/(1-(k-1)\mu)\nonumber.  
\end{align*}
The lemma follows from the definition of $\sigma_k$.
\end{IEEEproof}

\subsection{Proof of Lemma \ref{lem:cond}}
\begin{IEEEproof} 
For any $\Delta \in \C^{t\times n}$, $\langle L'+\Delta W^T, C'-\Delta \rangle$ is feasible to (\ref{eqn:opt}). Let $G$ be such that $\|G\|=1$, $\left\langle G, \P_{T'^{\perp}}(\Delta W^T) \right\rangle = \|\P_{T'^{\perp}}(\Delta W^T)\|_{*}$ and $\P_{T'}(G)=0$. Then $\P_{T'}(Q)+G$ is a subgradient of $\|L'\|_{*}$. Let $F$ be such that $F_i=-\Delta_i/\|\Delta_i\|_2$ if $i \in \bar{\I}^c$ and $\Delta_i \neq \bm0$, and $F_i=\bm0$ otherwise. Then $\P_{\bar{\I}} (QW^{\ddagger})/\lambda +F$ is a subgradient of $\|C'\|_{1,2}$. Then
\begin{equation}\label{eqn:worse}
\begin{aligned}
&\|L'+ \Delta W^T\|_{*}+\lambda \|C'-\Delta\|_{1,2}-\|L'\|_{*}-\lambda\|C'\|_{1,2} \\
\geq & \langle \P_{T'}(Q)+G, \Delta W^T \rangle -\lambda \langle \P_{\bar{\I}} (QW^{\ddagger})/\lambda +F, \Delta\rangle\\
=& \|\P_{T'^{\perp}}(\Delta W^T)\|_{*} + \lambda \|\P_{\bar{\I}^c}(\Delta)\|_{1,2} + \langle Q-\P_{T'^\perp}(Q), \Delta W^T \rangle \\ 
&-\langle QW^{\ddagger}-\P_{\bar{\I}^c} (QW^{\ddagger}), 
\Delta\rangle \\
\geq & (1-\|\P_{T'^\perp}(Q)\|)\|\P_{T'^{\perp}}(\Delta W^T)\|_{*} \\ &+(\lambda-\|(QW^{\ddagger})_{\bar{\I}^c}\|_{\infty,2} )\|\P_{\bar{\I}^c}(\Delta)\|_{1,2}\\
\geq & 0
\end{aligned}
\end{equation}
From (\ref{eqn:worse}), $\langle L', C'\rangle$ is an optimal solution to (\ref{eqn:opt}). If (\ref{eqn:cond}) holds with strict inequality, the last inequality of (\ref{eqn:worse}) is strict unless
\begin{equation}\label{eqn:temp1}
\|\P_{T'^{\perp}}(\Delta W^T)\|_{*}= \|\P_{\bar{\I}^c}(\Delta)\|_{1,2}=0.
\end{equation}
(\ref{eqn:temp1}) implies that $ \Delta W^T \in \P_{T'}$ and $\Delta \in \P_{\bar{\I}}$. Note that $\Delta \in \P_{\bar{\I}}$ implies that $\Delta W^T \in \P_{\bar{\J}}$.  Then
\begin{align}\label{eqn:med}
\P_{\bar{\J}}(\Delta W^T)&=\Delta W^T =\P_{T'}(\Delta W^T) \nonumber\\
&=\P_{U'}(\Delta W^T)+\P_{V'}\P_{U'^\perp}(\Delta W^T)\nonumber\\  
&= \P_{\bar{\J}}\P_{U'}(\Delta W^T)+\P_{V'}\P_{U'^\perp}(\Delta W^T),
\end{align}
where the last equality holds since $\P_{\bar{\J}}(\Delta W^T)=\Delta W^T$.
Thus, from (\ref{eqn:med}) we have $\P_{\bar{\J}}\P_{U'^{\perp}}(\Delta W^T)=\P_{V'}\P_{U'^\perp}(\Delta W^T)$,
which means $\P_{U'^\perp}(\Delta W^T) \in \P_{\bar{\J}} \cap \P_{V'}$. Then $\P_{U'^\perp}(\Delta W^T)$ is $0$ from the assumption. Then, $\P_{\bar{U}}(\Delta W^T)=\P_{U'}(\Delta W^T) =\Delta W^T$, where the first equality holds from (\ref{eqn:UU'}). Therefore,  for any optimal solution $\langle L'+\Delta W^T, C'-\Delta \rangle$ for some $\Delta \neq 0$ to (\ref{eqn:opt}),  $\Delta W^T \in \P_{\bar{U}}$, and $\Delta \in \P_{\bar{\I}}$. The claim follows.
\end{IEEEproof}

\subsection{Construction of $Q$}
Here we demonstrate that $Q$ in (\ref{eqn:Q}) is well defined. The key is to show (a) there exists $\hat{H} \in \G(C')$ such that (\ref{eqn:hath}) holds, and (b) the  infinite sum in (\ref{eqn:delta2}) converges. We prove these two properties through the following lemmas. 
\begin{lemma}\label{lem:convex}
There exists $\hat{H} \in \G(C')$ such that (\ref{eqn:hath}) holds.
\end{lemma}

\begin{proof}
Since $\langle L', C' \rangle$ is an optimal solution to the Oracle problem (\ref{eqn:oracle}), there exists $G',A' \in \C ^{t \times p}$, $B',Z \in \C^{t \times n}$, and some $\hat{H} \in  \G (C)$ such that
\begin{equation}\label{eqn:seq1}
(\bar{U}\hat{V}^\dag + G'+ \P_{\bar{U}^\perp}  (A'))W^{\ddagger}= \lambda (\hat{H}+Z)+\P_{\I^c}(B'),
\end{equation} 
where $\P_{T'^\perp}(G')=0$ and  $\P_{\I}(Z)=0$. Then 
\begin{equation}\label{eqn:seq2}
\P_{\bar{U}}\P_{\bar{\I}}(((\bar{U}\hat{V}^\dag + G'+ \P_{\bar{U}^\perp} (A'))W^{\ddagger}) %=\P_{\bar{U}}\P_{\bar{I}}(\bar{U}\hat{V}^\dag W)
=\bar{U}\hat{V}^\dag W^{\ddagger}_{\bar{\I}},
\end{equation} 
\begin{equation}\label{eqn:seq3}
\P_{\bar{U}}\P_{\bar{\I}}( \lambda (\hat{H}+Z)+\P_{\I^c}(B'))= \lambda \P_{\bar{U}}\P_{\bar{\I}}(\hat{H}) =\lambda\bar{U}\bar{U}^\dag \hat{H} \end{equation}
Combining (\ref{eqn:seq1})-(\ref{eqn:seq3}), we have
\begin{equation}\label{eqn:multiu}
\bar{U}\hat{V}^\dag W^{\ddagger}_{\bar{\I}} = \lambda\bar{U}\bar{U}^\dag \hat{H}.
\end{equation}
By multiplying $\bar{U}^\dag$ to both sides of (\ref{eqn:multiu}), we obtain Lemma \ref{lem:convex}.
\end{proof}

\begin{lemma}\label{lem:psi}
\begin{equation} \nonumber
\psi:=\|\P_{\hat{V}}\P_{W_{\bar{\I}}}\P_{\hat{V}}\| \leq \tilde{\psi}<1
\end{equation}
\end{lemma}

\begin{proof}
\begin{align}\nonumber
&\|\P_{\hat{V}}\P_{W_{\bar{\I}}}\P_{\hat{V}}(X)\| \\\nonumber
=&\|X\hat{V}\hat{V}^\dag W^{\ddagger}_{\bar{\I}}(W_{\bar{\I}}^T W^{\ddagger}_{\bar{\I}})^{-1}W_{\bar{\I}}^T \hat{V}\hat{V}^\dag\| \\\nonumber
\stackrel{(\rm{a})}{=}&\|X\hat{V}(\lambda\bar{U}^\dag \hat{H})(W_{\bar{\I}}^T W^{\ddagger}_{\bar{\I}})^{-1}(\lambda\bar{U}^\dag \hat{H})^\dag\hat{V}^\dag\| \\\nonumber
 \leq & \|X\|\|\hat{V}\bar{U}^\dag\| \|\lambda  \hat{H}\|\|(W_{\bar{\I}}^\dag W_{\bar{\I}})^{-1}\|\|\lambda\hat{H}^\dag \|\|\bar{U}\hat{V}^\dag\|\\\nonumber
 \stackrel{(\rm{b})}{\leq} & \|X\| \cdot 1 \cdot\lambda \sqrt{k} \cdot \sigma_k \cdot \lambda \sqrt{k}\cdot 1\\\nonumber
\stackrel{(\rm{c})}{\leq} & \|X\| \lambda_{\max} \tilde{k} \sigma_{\tilde{k}} \stackrel{(\rm{d})}{=} \|X\| \tilde{\psi},
\end{align}
where (a) follows from Lemma \ref{lem:convex}, (b) follows from the fact that $\hat{H}$ has at most $k$ nonzero columns with unit-norm, (c) follows from the property that $\lambda \leq \lambda_{\max}$, $k \leq \tilde{k}$ and $\sigma_{k} \leq \sigma_{\tilde{k}}$, and (d) follows from the definition of $\tilde{\psi}$. Then Lemma \ref{lem:psi} follows. 
\end{proof}

\begin{lemma}\label{lem:inverse}
$\P_{\hat{V}} (I-\P_{W_{\bar{\I}}})\P_{\hat{V}}$ is an injection from $\P_{\hat{V}}$ to  $\P_{\hat{V}}$,  and its inverse operation is $(I+\sum_{i=1}^\infty (\P_{\hat{V}}\P_{W_{\bar{\I}}}\P_{\hat{V}})^i)$.
\end{lemma}
\begin{proof}
Since  $\|\P_{\hat{V}}\P_{W_{\bar{\I}}}\P_{\hat{V}}\| <1$ from Lemma \ref{lem:psi}, then  $(I+\sum_{i=1}^\infty (\P_{\hat{V}}\P_{W_{\bar{\I}}}\P_{\hat{V}})^i)$ is well defined. For any $X \in \P_{\hat{V}}$, we have 
\begin{align}
&\P_{\hat{V}}(I-\P_{W_{\bar{\I}}})\P_{\hat{V}}(I+\sum_{i=1}^\infty (\P_{\hat{V}}\P_{W_{\bar{\I}}}\P_{\hat{V}})^i)(X)\nonumber\\=&\P_{\hat{V}}(I-\P_{\hat{V}}\P_{W_{\bar{\I}}}\P_{\hat{V}})(I+\sum_{i=1}^\infty (\P_{\hat{V}}\P_{W_{\bar{\I}}}\P_{\hat{V}})^i)(X)\nonumber\\
=& \P_{\hat{V}}(X)=X.
\end{align}
Then the lemma follows.
\end{proof}

\subsection{Proof of Lemma \ref{lem:certificate}}
\begin{proof}
We need to show that $Q$ defined in (\ref{eqn:Q}) satisfies all the conditions in (\ref{eqn:cond}).  We first summarize some properties that will be used in the proof. Since $W$ has unit-norm columns, $|W^\dag _i W_j| \leq \mu$ for all $i \neq j$, and $|\bar{\I}|\leq k$,  we have
\begin{equation}\label{eqn:WI}
\|W_{\bar{\I}}\|=\sqrt{\lambda_{\max}(W_{\bar{\I}}^\dag W_{\bar{\I}})} \leq \sqrt{1+(k-1)\mu},
\end{equation}
where the inequality follows from the Ger\v{s}hgorin circle theorem. 
From $|\bar{\I}|\leq k$ and $|W^\dag _i W_j| \leq \mu$ for all $i \neq j$, we have $\|(W_{\bar{\I}}^\dag W_{\bar{\I}^c})\|_{\infty,2} \leq \sqrt{k} \mu$. Since $\hat{H}$ has at most $k$ unit-norm columns while other columns are zero, we have 
\begin{equation}\label{eqn:lambdaH}
\|\lambda \hat{H} \| \leq \lambda \sqrt{k}.
\end{equation}

Step 1: verification of (a) of (\ref{eqn:cond}).
\begin{equation}\label{eqn:U}
\P_{U'}(Q)\stackrel{(\rm{a})}{=}\P_{\bar{U}}(Q)= \bar{U}\hat{V}^\dagger+\P_{\bar{U}}(\Phi)-\P_{\bar{U}}(\Phi)-0= \bar{U}\hat{V}^\dagger,
\end{equation}
where (a) follows from  (\ref{eqn:space}). From (\ref{eqn:UU'}), we have 
\begin{equation} \nonumber
\hat{V}\hat{V}^\dag =V' U'^\dag \bar{U} \bar{U}^\dag U' V'^\dagger\stackrel{(\rm{b})}{=}V' U'^\dag U'U'^\dag U' V'^\dagger=V'V'^\dagger,
\end{equation}
where (b) follows from (\ref{eqn:UU'}). Thus, $\P_{V'}(\cdot)=\P_{\hat{V}}(\cdot)$. Then
\begin{align}\label{eqn:V}
\P_{V'}(Q)&=\P_{\hat{V}}(Q)\stackrel{(\rm{c})}{=} \bar{U}\hat{V}^\dagger+\P_{\hat{V}}(\Phi)-\P_{\hat{V}}\P_{\bar{U}}(\Phi) \nonumber\\&-  \P_{\hat{V}} (I-\P_{W_{\bar{\I}}})\P_{\hat{V}}(I+\sum_{i=1}^\infty (\P_{\hat{V}}\P_{W_{\bar{\I}}}\nonumber\\
& \P_{\hat{V}})^i)\P_{\hat{V}} \P_{\bar{U}^\perp}(\Phi)\nonumber\\
&\stackrel{(\rm{d})}{=} \bar{U}\hat{V}^\dagger+\P_{\hat{V}}(\Phi)-\P_{\hat{V}}\P_{\bar{U}}(\Phi)-\P_{\hat{V}} \P_{\bar{U}^\perp}(\Phi) \nonumber\\
&=\bar{U}\hat{V}^\dagger. 
\end{align}
(c) follows since $\P_{W_{\bar{\I}}}$, $\P_{\hat{V}}$, and $\P_{\hat{V}}\P_{W_{\bar{\I}}}\P_{\hat{V}}$ are all given by right matrix multiplication, while $\P_{\bar{U}^\perp}$
is given by left matrix
multiplication. (d) follows from Lemma \ref{lem:inverse}. Combining (\ref{eqn:U}) and (\ref{eqn:V}), we obtain that (a) of (\ref{eqn:cond}) holds.

Step 2:  verification of (b) of (\ref{eqn:cond}).
\begin{align}\label{eqn:stp2}
&\|\P_{T'^{\perp}}(Q)\|=\|\P_{\hat{V}^{\perp}}\P_{\bar{U}^{\perp}}(\Phi)-\nonumber \\
& \quad \P_{\bar{U}^{\perp}}\P_{\hat{V}^\perp}(I-\P_{W_{\bar{\I}}})\P_{\hat{V}}(I+\sum_{i=1}^\infty (\P_{\hat{V}}\P_{W_{\bar{\I}}}\P_{\hat{V}})^i)\P_{\hat{V}}(\Phi)\|\nonumber \\
\leq  & \|\Phi\|+  (1+\sum_{i=1}^{\infty}\psi^i)\|\Phi\| = \frac{2-\psi}{1-\psi}  \|\Phi\|\nonumber\\
\stackrel{(\rm{e})}{\leq} & \frac{2-\psi}{1-\psi} \|\lambda\hat{H} \|\|(W_{\bar{\I}}^\dag W_{\bar{\I}})^{-1}\| \|W_{\bar{\I}}^T\| \nonumber\\
\stackrel{(\rm{f})}{\leq} &  \frac{2-\psi}{1-\psi}\lambda \sqrt{k} \sigma_k \sqrt{1+(k-1)\mu}\\
\stackrel{(\rm{g})}{\leq} &  \frac{2-\tilde{\psi}}{1-\tilde{\psi}} \sqrt{\frac{\tilde{\psi}}{\tilde{k}\sigma_{\tilde{k}}}} \sqrt{\tilde{k}}  \sigma_{\tilde{k}} \sqrt{1+(\tilde{k}-1)\mu}\\
\stackrel{(\rm{h})}{\leq}& \frac{2-\tilde{\psi}}{1-\tilde{\psi}} \sqrt{\tilde{\psi}}  \sqrt{\frac{1+(\tilde{k}-1)\mu}{1-(\tilde{k}-1)\mu}}\\
\stackrel{(\rm{i})}{\leq} &  \frac{2-\tilde{\psi}}{1-\tilde{\psi}} \sqrt{\tilde{\psi}} \sqrt{\frac{1+c}{1-c}} \stackrel{(\rm{j})}{\leq} 1 \nonumber.  
\end{align}
where (e) follows from the definition of $\Phi$, and (f) follows from (\ref{eqn:WI}) and (\ref{eqn:lambdaH}). (g) follows from the property that $\psi \leq \tilde{\psi}$, $1 \leq k \leq \tilde{k}$, $\lambda \leq \lambda_{\max,\tilde{k}}$, and $\sigma_k \leq \sigma_{\tilde{k}}$. (h) follows from Lemma \ref{lem:sigmak}. (i) follow from  $\tilde{k}\mu \leq c$, and (j) follows from (\ref{eqn:psistar}). Then (b) of (\ref{eqn:cond}) holds.

Step 3: verification of (c) of (\ref{eqn:cond}). First consider
\begin{align}
&(\Delta_2W^{\ddagger})_{\bar{\I}}\nonumber\\
=& (\P_{\bar{U}^\perp}(I-\P_{W_{\bar{\I}}})\P_{\hat{V}}(I+\sum_{i=1}^\infty (\P_{\hat{V}}\P_{W_{\bar{\I}}}\P_{\hat{V}})^i)\P_{\hat{V}}(\Phi) W^{\ddagger})_{\bar{\I}}\nonumber\\
\stackrel{(\rm{k})}{=}& (\P_{\bar{U}^\perp}\P_{\hat{V}}(I+\sum_{i=1}^\infty (\P_{\hat{V}}\P_{W_{\bar{\I}}}\P_{\hat{V}})^i)\P_{\hat{V}}(\Phi))(I-\nonumber\\ & W^{\ddagger}_{\bar{\I}}(W_{\bar{\I}}^T W^{\ddagger}_{\bar{\I}})^{-1}W_{\bar{\I}}^T)W^{\ddagger}_{\bar{\I}} =0 \nonumber
\end{align}
where (k) holds since  $\P_{W_{\bar{\I}}}$, $\P_{\hat{V}}$, and $\P_{\hat{V}}\P_{W_{\bar{\I}}}\P_{\hat{V}}$ are all given by right matrix multiplication, while $\P_{\bar{U}^\perp}$ is given by left matrix multiplication. Then
\begin{align}
(QW^{\ddagger})_{\bar{\I}}&=(\bar{U}\hat{V}W^{\ddagger}+\Phi W^{\ddagger} -\P_{\bar{U}}(\Phi)W^{\ddagger})_{\bar{\I}}-(\Delta_2W^{\ddagger})_{\bar{\I}}\nonumber \\
&= \bar{U}\hat{V} W^{\ddagger}_{\bar{\I}}+\Phi W^{\ddagger}_{\bar{\I}}-\P_{\bar{U}}(\Phi)W^{\ddagger}_{\bar{\I}}-0\nonumber\\
&\stackrel{(\rm{l})}{=} \lambda \bar{U}\bar{U}^\dag \hat{H}+\lambda \hat{H}- \lambda \bar{U}\bar{U}^\dag \hat{H}\nonumber\\
&=\lambda \hat{H} \in \lambda \G(C'),
\end{align}
where (l) follows from Lemma \ref{lem:convex} and the definition of $\Phi$ in (\ref{eqn:phi}). Then (c) of (\ref{eqn:cond}) holds.

Step 4: verification of (d) of (\ref{eqn:cond}). First consider
\begin{align}
&\|(\Delta_2W^{\ddagger})_{\bar{\I}^c}\|_{\infty,2}\nonumber\\
=& \|\P_{\bar{U}^\perp}\P_{\hat{V}}(I+\sum_{i=1}^\infty (\P_{\hat{V}}\P_{W_{\bar{\I}}}\P_{\hat{V}})^i)\nonumber \\
& \cdot \Phi\hat{V}\hat{V}^\dag(I- W^{\ddagger}_{\bar{\I}}(W_{\bar{\I}}^T W^{\ddagger}_{\bar{\I}})^{-1}W_{\bar{\I}}^T  )W^{\ddagger}_{\bar{\I}^c}\|_{\infty,2}\nonumber\\
=& \|\P_{\bar{U}^\perp}\P_{\hat{V}}(I+\sum_{i=1}^\infty (\P_{\hat{V}}\P_{W_{\bar{\I}}}\P_{\hat{V}})^i)\Phi(\hat{V}\hat{V}^\dag W^{\ddagger}_{\bar{\I}^c}- \nonumber\\& \hat{V}\hat{V}^\dag  W^{\ddagger}_{\bar{\I}}(W_{\bar{\I}}^T W^{\ddagger}_{\bar{\I}})^{-1}W_{\bar{\I}}^T  )W^{\ddagger}_{\bar{\I}^c}\|_{\infty,2}\nonumber\\
\leq & \|I+\sum_{i=1}^\infty (\P_{\hat{V}}\P_{W_{\bar{\I}}}\P_{\hat{V}})^i\|\|\Phi\|\Big(\|\hat{V}\|  \|\hat{V} W^{\ddagger}_{\bar{\I}^c}\|_{\infty,2}\nonumber\\
&+ \| \hat{V}\|\hat{V}^\dag W^{\ddagger}_{\bar{\I}}\|\|(W_{\bar{\I}}^\dag W_{\bar{\I}})^{-1}\|\|W_{\bar{\I}}^T  W^{\ddagger}_{\bar{\I}^c}\|_{\infty,2}\Big)\nonumber\\
\leq & \frac{\|\Phi\|(\epsilon + \lambda \sqrt{k}\sigma_{k} \sqrt{k}\mu)}{1-\psi}\leq \frac{\epsilon + \lambda k \sigma_{k} \mu}{2-\psi} \leq \frac{\epsilon + \lambda k\sigma_{k}\mu}{2-\tilde{\psi}}, \nonumber
\end{align}
where the second to last inequality follows from (e) to (j) in step 2.
\begin{align}
&\|(QW^{\ddagger})_{\bar{\I}^c}\|_{\infty,2}\nonumber\\
=&\|(\bar{U}\hat{V}W^{\ddagger}+\Phi W^{\ddagger} -\P_{\bar{U}}(\Phi)W^{\ddagger}-\Delta_2W^{\ddagger})_{\bar{\I}^c}\|_{\infty,2} \nonumber\\
=& \|\bar{U}\hat{V} W^{\ddagger}_{\bar{\I}^c}+\P_{\bar{U}^\perp}(\Phi)W^{\ddagger}_{\bar{\I}^c}-(\Delta_2W^{\ddagger})_{\bar{\I}^c}\|_{\infty,2}\nonumber\\
\leq &\|\bar{U}\hat{V} W^{\ddagger}_{\bar{\I}^c}\|_{\infty,2}+ \| (I-\bar{U}\bar{U})^\dag \lambda \hat{H} (W_{\bar{\I}}^T W^{\ddagger}_{\bar{\I}})^{-1}W_{\bar{\I}}^T W^{\ddagger}_{\bar{\I}^c} \|_{\infty,2}\nonumber\\
& +\|(\Delta_2W^{\ddagger})_{\bar{\I}^c}\|_{\infty,2}\nonumber\\
\leq &\|\bar{U}\| \|\hat{V} W^{\ddagger}_{\bar{\I}^c}\|_{\infty,2} + \nonumber\\
& \|(I- \bar{U}\bar{U})^\dag \|\|\lambda \hat{H}\|\|(W_{\bar{\I}}^\dag W_{\bar{\I}})^{-1}\|\| W_{\bar{\I}}^T W^{\ddagger}_{\bar{\I}^c} \|_{\infty,2} +\nonumber\\
& \|(\Delta_2W^{\ddagger})_{\bar{\I}^c}\|_{\infty,2}\nonumber\\
\leq & \epsilon+ \lambda \sqrt{k} \sigma_k \sqrt{k}\mu +\nonumber    \\
& \frac{\lambda \sigma_{k} \sqrt{k+(k^2-k)\mu} (\epsilon + \sigma_{k} \mu\sqrt{k+(k^2-k)\mu})}{1-\psi}\nonumber\\
\leq & (1+\frac{1}{2-\tilde{\psi}})(\epsilon+ \lambda k \sigma_k \mu)  , \nonumber\\
\leq & (1+\frac{1}{2-\tilde{\psi}})(\epsilon+ \lambda \tilde{k} \sigma_{\tilde{k}} \mu), \label{eqn:stp4}\\
\leq & \lambda\nonumber,
\end{align}
where the last inequality follows from $\lambda \geq \lambda_{\min,\tilde{k}}$. Then  (d) of (\ref{eqn:cond}) holds.
\end{proof}

\subsection{Proof of Lemma \ref{lem:cond1}}
\begin{proof}
We define 
\begin{equation*}
\tilde{C} = \bar{C}+\P_{\bar{\I}}\P_{\bar{U}}(C^*-\bar{C}) \text{ and } \tilde{L} = \bar{L}-\P_{\bar{\I}}\P_{\bar{U}}(C^*-\bar{C})W^T.
\end{equation*}
Note that $\P_{\bar{U}}(\tilde{L})=\tilde{L}$, $\P_{\bar{\I}}(\tilde{C})=\tilde{C}$ and $\bar{L}+\bar{C}W^T=\tilde{L}+\tilde{C}W^T$. We further define $N_L=L^*-\bar{L}$, $N_C=C^*-\bar{C}$, and $N^+_C = C^*-\tilde{C}$. Note that $\P_{\bar{\I}^c}(N_C^+)=\P_{\bar{\I}^c}(N_C)$ from the definition of $N^+_C$.  Let $E=N_L+N_CW^T$. We have
\begin{align}\label{eqn:E}
&\|E\|_F = \|L^*+C^*W^T-(\bar{L}+\bar{C}W^T)\|_F \nonumber\\
\le & \|L^*+C^*W^T-M\|_F+\|N\|_F \le 2\eta,
\end{align}
where the last inequality holds since ($L^*$, $C^*$) is the solution to (\ref{eqn:opt}) and $\|N\|_F \le \eta$. Let $G$ be such that $\|G\|=1$, $\left\langle G, \P_{T^{*{\perp}}}(\Delta W^T) \right\rangle = \|\P_{T^{*{\perp}}}(\Delta W^T)\|_{*}$ and $\P_{T^*}(G)=0$. Let $F$ be such that $F_i=\Delta_i/\|\Delta_i\|_2$ if $i \in \bar{\I}$ and $\Delta_i \neq 0$, and $F_i=0$ otherwise. Then
\begin{equation}\label{eqn:worse1}
\begin{aligned}
&\|\bar{L}\|_*+\lambda\|\bar{C}\|_{1,2} \stackrel{(\rm{m})}{\ge} \|L^*\|_*+\lambda\|C^*\|_{1,2} \\
\stackrel{(\rm{n})}{\ge} & \|\bar{L}\|_*+\lambda\|\bar{C}\|_{1,2}+\langle \P_{\bar{T}}(Q)+G,N_L\rangle+\lambda\langle \P_{\bar{\I}}(QW^{\ddagger})/\lambda \\
&+F,N_C\rangle \\
= & \|\bar{L}\|_*+\lambda\|\bar{C}\|_{1,2}+\|\P_{\bar{T}^{\perp}}(N_L)\|_*+\langle \P_{\bar{T}}(Q),N_L \rangle\\   
&+\lambda\|\P_{\bar{\I}^c}(N_C)\|_{1,2}+\langle \P_{\bar{\I}}(QW^{\ddagger}),N_C \rangle \\
= & \|\bar{L}\|_*+\lambda\|\bar{C}\|_{1,2}+\|\P_{\bar{T}^{\perp}}(N_L)\|_*+\lambda\|\P_{\bar{\I}^c}(N_C)\|_{1,2} \\
&-\langle \P_{\bar{T}^{\perp}}(Q),N_L \rangle-\langle \P_{\bar{\I}^c}(QW^{\ddagger}),N_C \rangle+ \langle Q,N_L+N_CW^T \rangle \\
\ge &  \|\bar{L}\|_*+\lambda\|\bar{C}\|_{1,2}+(1-\|\P_{\bar{T}^{\perp}}(Q)\|)\|\P_{\bar{T}^{\perp}}(N_L)\|_* \\
&+(\lambda-\|\P_{\bar{\I}^c}(QW^{\ddagger})\|_{\infty,2})\|\P_{\bar{\I}^c}(N_C)\|_{1,2}+\langle Q,E \rangle \\
\ge &  \|\bar{L}\|_*+\lambda\|\bar{C}\|_{1,2}+\frac{1}{2}\|\P_{\bar{T}^{\perp}}(N_L)\|_*+\frac{\lambda}{2}\|\P_{\bar{\I}^c}(N_C)\|_{1,2}\\
&-2\eta\|Q\|_F,
\end{aligned}
\end{equation}
where (m) holds because of the optimality of ($L^*$, $C^*$) and (n) holds because of the convexity of the objective function of (\ref{eqn:opt}). We can see that the last inequality of (\ref{eqn:worse1}) follows from (b) and (d) of (\ref{eqn:cond1}). Then we have
\begin{equation}\label{eqn:inequa1}
\frac{1}{2}\|\P_{\bar{T}^{\perp}}(N_L)\|_*+\frac{\lambda}{2}\|\P_{\bar{\I}^c}(N_C)\|_{1,2}-2\eta\|Q\|_F \le 0.
\end{equation}
Note that 
\begin{align}
& \|Q\|_F = \|\P_{\bar{T}}(Q)+\P_{\bar{T}^{\perp}}(Q)\|_F \nonumber \\ 
= & \sqrt{\|\P_{\bar{T}}(Q)\|^2_F+\|\P_{\bar{T}^{\perp}}(Q)\|^2_F} \nonumber\\
= & \sqrt{\|\bar{U} \bar{V}^{\dagger}\|^2_F+\|\P_{\bar{T}^{\perp}}(Q)\|^2_F} \stackrel{(\rm{o})}{\leq} \frac{1}{2}\sqrt{\min(t,p)+3r}, \label{eqn:QF}
\end{align}
where the last equality follows from (a) of (\ref{eqn:cond1}). 
The inequality (o) holds from $\|\bar{U} \bar{V}^{\dagger}\|_F = \sqrt{\rm{trace}(\bar{V}\bar{U}^{\dagger}\bar{U}\bar{V}^{\dagger})}=\sqrt{r}$,
and
\begin{equation}\nonumber
\|\P_{\bar{T}^{\perp}}(Q)\|_F \le \textrm{rank}(\P_{\bar{T}^{\perp}}(Q))\cdot\|\P_{\bar{T}^{\perp}}(Q)\|\leq \frac{\sqrt{\min(t,p)-r}}{2}.
\end{equation}
Since $\theta  = \min(t,p)$, combining (\ref{eqn:inequa1}) and (\ref{eqn:QF}), we have 
\begin{equation}\label{eqn:PNL}
\|\P_{\bar{T}^{\perp}}(N_L)\|_F \le \|\P_{\bar{T}^{\perp}}(N_L)\|_* \le 2\eta\sqrt{\theta+3r},
\end{equation}
\begin{equation}\label{eqn:PNC}
\|\P_{\bar{\I}^c}(N_C)\|_F \le \|\P_{\bar{\I}^c}(N_C)\|_{1,2} \le \frac{2}{\lambda}\eta\sqrt{\theta+3r}.
\end{equation}
	
From the definition of $\P_{W_{\bar{\I}}}$ in (\ref{eqn:PWI}), one can check that
\begin{equation}\label{eqn:PWIWI}
 \P_{W_{\bar{\I}}}(\P_{\bar{\I}}(W)^T)=  \P_{\bar{\I}}(W)^T.
\end{equation}
Then we have
\begin{equation}\label{eqmain}
\begin{aligned}
& \P_{\bar{\I}}(N_C^+)W^T =  \P_{\bar{\I}}(N_C^+)\P_{\bar{\I}}(W)^T\\
= & \P_{\bar{\I}}(N_C^+)\P_{W_{\bar{\I}}}(\P_{\bar{\I}}(W)^T)=\P_{\bar{\I}}(N_C^+)\P_{W_{\bar{\I}}}(W^T)\\
= & \P_{W_{\bar{\I}}}(N_C^+W^T-\P_{\bar{\I}^c}(N_C^+)W^T)\\ 
\stackrel{(\rm{p})}{=} & \P_{W_{\bar{\I}}}(E-\P_{\bar{T}^{\perp}}(N_L)-\P_{\bar{T}}(N_L)-\P_{\bar{\I}}\P_{\bar{U}}(N_C)W^T \\
& -\P_{\bar{\I}^c}(N_C^+)W^T)\\ 
\stackrel{(\rm{q})}{=} & \P_{W_{\bar{\I}}}(E-\P_{\bar{T}^{\perp}}(N_L)-\P_{\bar{T}}(E)+\P_{\bar{T}}(N_CW^T)\\
& -\P_{\bar{\I}}\P_{\bar{U}}(N_C)W^T-\P_{\bar{\I}^c}(N_C^+)W^T) \\
= & \P_{W_{\bar{\I}}}(\P_{\bar{T}^{\perp}}(E)-\P_{\bar{T}^{\perp}}(N_L)-\P_{\bar{\I}^c}(N_C)W^T\\
& +\P_{\bar{T}}(\P_{\bar{\I}}(N_C)W^T)+\P_{\bar{T}}(\P_{\bar{\I}^c}(N_C)W^T)\\ &-\P_{\bar{\I}}\P_{\bar{U}}(N_C)W^T)\\
\stackrel{(\rm{r})}{=} &
\P_{W_{\bar{\I}}}(\P_{\bar{T}^{\perp}}(E)-\P_{\bar{T}^{\perp}}(N_L)-\P_{\bar{\I}^c}(N_C)W^T+\\
& \P_{\bar{T}}(\P_{\bar{\I}^c}(N_C)W^T)+\P_{\bar{U}}(\P_{\bar{\I}}(N_C)W^T)+\\
&\P_{\bar{V}}(\P_{\bar{\I}}(N_C)\P_{\bar{\I}}(W)^T)-\P_{\bar{U}}\P_{\bar{V}}(\P_{\bar{\I}}(N_C)W^T)\\
&-\P_{\bar{\I}}\P_{\bar{U}}(N_C)W^T)\\ 
\stackrel{(\rm{s})}{=} &
\P_{W_{\bar{\I}}}(\P_{\bar{T}^{\perp}}(E)-\P_{\bar{T}^{\perp}}(N_L)-\P_{\bar{\I}^c}(N_C)W^T+\\
& \P_{\bar{T}}(\P_{\bar{\I}^c}(N_C)W^T)+\P_{\bar{V}}(N_C\P_{\bar{\I}}(W)^T)-\P_{\bar{V}}(\P_{\bar{\I}^c}(N_C)\\
& \P_{\bar{\I}}(W)^T) - \P_{\bar{U}}\P_{\bar{V}}(\P_{\bar{\I}}(N_C)W^T))\\ 
\stackrel{(\rm{t})}{=} & \P_{W_{\bar{\I}}}(\P_{\bar{T}^{\perp}}(E)-\P_{\bar{T}^{\perp}}(N_L)-\P_{\bar{\I}^c}(N_C)W^T+\\
& \P_{\bar{T}}(\P_{\bar{\I}^c}(N_C)W^T)+\P_{\bar{V}}(N_C^+\P_{\bar{\I}}(W)^T)). 
\end{aligned}
\end{equation}
where (p) and (q) follow from the definition $E=N_L+N_CW^T$ and $N^+_C = N_C - \P_{\bar{\I}}\P_{\bar{U}}(N_C)$. (r) follows the definition of $\P_{\bar{T}}$. (s) holds because $\P_{\bar{U}}(\P_{\bar{\I}}(N_C)W^T)=\P_{\bar{\I}}\P_{\bar{U}}(N_C)W^T$. (t) holds because of the equality (\ref{eqn:sec2}) shown as follows:
\begin{equation}\label{eqn:sec2}
\begin{aligned}
& \P_{\bar{V}}(N_C\P_{\bar{\I}}(W)^T)-\P_{\bar{U}}\P_{\bar{V}}(\P_{\bar{\I}}(N_C)\P_{\bar{\I}}(W)^T) \\
= & \P_{\bar{V}}(N_C\P_{\bar{\I}}(W)^T-\P_{\bar{\I}}\P_{\bar{U}}(N_C)\P_{\bar{\I}}(W)^T) \\
= & \P_{\bar{V}}(N_C^+\P_{\bar{\I}}(W)^T)
\end{aligned}
\end{equation}
Note that
\begin{align}
& \|\P_{W_{\bar{\I}}}\P_{\bar{V}}(N_C^+\P_{\bar{\I}}(W)^T)\|_F\nonumber\\
= & \|\P_{W_{\bar{\I}}}\P_{\bar{V}}(N_C^+\P_{W_{\bar{\I}}}\P_{\bar{\I}}(W)^T)\|_F\nonumber\\
= & \|N_C^+\P_{\bar{\I}}(W)^TW^{\ddagger}_{\bar{\I}}(W_{\bar{\I}}^{T}W^{\ddagger}_{\bar{\I}})^{-1}W_{\bar{\I}}^{T}\bar{V}\bar{V}^{\dagger}W^{\ddagger}_{\bar{\I}}(W_{\bar{\I}}^{T}W^{\ddagger}_{\bar{\I}})^{-1}W_{\bar{\I}}^{T}\|_F\nonumber\\
\stackrel{(\rm{u})}{\leq} & \|N_C^+\P_{\bar{\I}}(W)^T\|_F\| \bar{V}^{\dagger}W^{\ddagger}_{\bar{\I}}(W_{\bar{\I}}^{T}W^{\ddagger}_{\bar{\I}})^{-1}W_{\bar{\I}}^{T}\bar{V}\|\nonumber\\
=   & \|N_C^+\P_{\bar{\I}}(W)^T\|_F\|\bar{V}\bar{V}^{\dagger}W^{\ddagger}_{\bar{\I}}(W_{\bar{\I}}^{T}W^{\ddagger}_{\bar{\I}})^{-1}W_{\bar{\I}}^{T}\bar{V}\bar{V}^{\dagger}\|\nonumber\\
= & \psi\|\P_{\bar{\I}}(N_C^+)W^T\|_F \leq \tilde{\psi}\|\P_{\bar{\I}}(N_C^+)W^T\|_F, \nonumber
\end{align}
where the first equality holds from (\ref{eqn:PWIWI}), and (u) holds because $\|AB\|_F \leq \|A\|_F\|B\|$ and $\|A^\dag A\|=\|AA^\dag\|$ for matrices $A$ and $B$. From (\ref{eqmain}), we have
\begin{align}
& \|\P_{\bar{\I}}(N_C^+)W^T\|_F \nonumber \\
\le & (\|\P_{\bar{T}^{\perp}}(E)\|_F+\|\P_{\bar{T}^{\perp}}(N_L)\|_F+\|\P_{\bar{T}^{\perp}}(\P_{\bar{\I}^c}(N_C)W^T)\|_F) \nonumber\\
& \|W^{\ddagger}_{\bar{\I}}(W_{\bar{\I}}^{T}W^{\ddagger}_{\bar{\I}})^{-1}W_{\bar{\I}}^{T}\|+\tilde{\psi}\|\P_{\bar{\I}}(N_C^+)W^T\|_F \nonumber\\
\le & \|E\|_F+\|\P_{\bar{T}^{\perp}}(N_L)\|_F+\|\P_{\bar{\I}^c}(N_C)\|_F\|W\| +\nonumber\\
& \tilde{\psi}\|\P_{\bar{\I}}(N_C^+)W^T\|_F, \label{eqn:PINCW}
\end{align} 
where the last inequality uses the property that 
$\|W^{\ddagger}_{\bar{\I}}(W_{\bar{\I}}^{T}W^{\ddagger}_{\bar{\I}})^{-1}W_{\bar{\I}}^{T}\|=1$. From similar arguments as in (\ref{eqn:WI}),  we have $\|W\| \le \sqrt{1+(n-1)\mu}$. Then combining (\ref{eqn:E}), (\ref{eqn:PNL}), (\ref{eqn:PNC}), and (\ref{eqn:PINCW}), we obtain
\begin{align}\label{eqn:PNCW} 
& \|\P_{\bar{\I}}(N_C^+)W^T\|_F \le (1+\frac{\lambda+\sqrt{1+(n-1)\mu}}{\lambda}\sqrt{\theta+3r})\frac{2\eta}{1-\tilde{\psi}}.
\end{align}
	
Furthermore,	
\begin{equation}\label{eqn:ineq1}\nonumber
\begin{aligned}
& \|\P_{\bar{\I}}(N^+_C)\|_F = \|\P_{\bar{\I}}(N^+_C)W^TW^{\ddagger}_{\bar{\I}}(W^T_{\bar{\I}}W^{\ddagger}_{\bar{\I}})^{-1}\|_F\\ 
\le & \|\P_{\bar{\I}}(N^+_C)W^T\|_F\|W^{\ddagger}_{\bar{\I}}\|\|(W^T_{\bar{\I}}W^{\ddagger}_{\bar{\I}})^{-1}\|\\
\le & (1+\frac{\lambda+\sqrt{1+(n-1)\mu}}{\lambda}\sqrt{\theta+3r})\frac{2\eta\sigma_k\sqrt{1+(k-1)\mu}}{1-\tilde{\psi}},
\end{aligned}
\end{equation}
where the last inequality follows from 	(\ref{eqn:PNCW}), (\ref{eqn:WI}), and (\ref{eqn:sigmak}). 
We also have
\begin{equation}\label{eqn:ineq2}\nonumber
\begin{aligned}
& \|N^+_CW^T\|_F = \|\P_{\bar{\I}^c}(N_C)W^T+\P_{\bar{\I}}(N^+_C)W^T\|_F\\
\le & \|\P_{\bar{\I}^c}(N_C)W^T\|_F+\|\P_{\bar{\I}}(N^+_C)W^T\|_F\\
\le & \|\P_{\bar{\I}^c}(N_C)\|_F\|W\|+\|\P_{\bar{\I}}(N^+_C)W^T\|_F\\
\le & (1+\frac{\lambda+(2-\tilde{\psi})\sqrt{1+(n-1)\mu}}{\lambda}\sqrt{\theta+3r})\frac{2\eta}{1-\tilde{\psi}}.
\end{aligned}
\end{equation}
	
Finally, we have 	
\begin{align}
& \|C^*-\tilde{C}\|_F = \|\P_{\bar{\I}^c}(N_C)+\P_{\bar{\I}}(N^+_C)\|_F\nonumber\\
\le & \|\P_{\bar{\I}^c}(N_C)\|_F+\|\P_{\bar{\I}}(N^+_C)\|_F\nonumber\\
\le & (1+(\frac{\lambda+\sqrt{1+(n-1)\mu}}{\lambda}+\frac{1-\tilde{\psi}}{\lambda\sigma_k\sqrt{1+(k-1)\mu}})\sqrt{\theta+3r})\nonumber\\
& \frac{2\eta\sigma_k\sqrt{1+(k-1)\mu}}{1-\tilde{\psi}},\nonumber
\end{align}
and
\begin{align}
& \|L^*-\tilde{L}\|_F
=  \|L^*-\bar{L}+\tilde{C}W^T-\bar{C}W^T\|_F \nonumber\\
= & \|L^*-\bar{L}+C^*W^T-\bar{C}W^T+\tilde{C}W^T-C^*W^T\|_F\nonumber\\
%\le & \|E\|_F+\|(C^*-\bar{C})W^T\|_F\\
= & \|E-N^+_CW^T\|_F\nonumber 
\le   \|E\|_F+\|N^+_CW^T\|_F\nonumber\\
\le & (2-\tilde{\psi}+\frac{\lambda+(2-\tilde{\psi})\sqrt{1+(n-1)\mu}}{\lambda}\sqrt{\theta+3r})\frac{2\eta}{1-\tilde{\psi}}.\nonumber
\end{align}
\end{proof}

\subsection{Proof of Lemma \ref{lem:certificate1}}

\begin{proof}
Since equalities (a) and (c) of (\ref{eqn:cond1}) are the same as those in (\ref{eqn:cond}) and the construction of $Q$ remains the same, then (a) and (c) of (\ref{eqn:cond1}) have been proved in step 1 and 3 of the proof of Lemma \ref{lem:certificate}. We only need to show that (b) and (d) hold when $\lambda$ belongs to [$\lambda'_{\text{min},\tilde{k}}$, $\lambda'_{\text{max},\tilde{k}}$]. From (\ref{eqn:stp2}), that is proved in the proof of Lemma \ref{lem:certificate}, and $\lambda \leq \lambda'_{\text{max},\tilde{k}}$, we have
\begin{equation}\nonumber
\|\P_{T'^{\perp}}(Q)\| \leq \frac{2-\tilde{\psi}}{1-\tilde{\psi}}\lambda  \sigma_{\tilde{k}} \sqrt{\tilde{k}+(\tilde{k}^2-\tilde{k})\mu} \leq \frac{1}{2}.
\end{equation}
From (\ref{eqn:stp4}) and $\lambda \geq \lambda'_{\text{min},\tilde{k}}$, we have
\begin{equation}\nonumber
\|(QW^{\ddagger})_{\bar{\I}^c}\|_{\infty,2} \leq 
(1+\frac{1}{2-\tilde{\psi}})(\epsilon+ \lambda \tilde{k} \sigma_{\tilde{k}} \mu) \leq \frac{\lambda}{2}.
\end{equation}	
\end{proof}

\begin{IEEEbiographynophoto}
{Pengzhi Gao} (S'14) received the B.E. degree   from Xidian University, Xian, China, in 2011 and the M.S. degree in electrical engineering from University of Pennsylvania, Philadelphia, PA, in 2013. 

He is  pursuing the Ph.D. degree in electrical engineering at Rensselaer Polytechnic Institute, Troy, NY.
His research interests include signal processing, compressive sensing, low-rank matrix recovery, and power networks.
\end{IEEEbiographynophoto}

\begin{IEEEbiographynophoto}
{Meng Wang} (M'12) received   the Ph.D. degree from Cornell University, Ithaca, NY, USA, in  2012. %, respectively.

She is    an Assistant Professor in   the department of  Electrical, Computer, and Systems Engineering at Rensselaer Polytechnic Institute. Her research interests include high dimensional data analysis and their applications in power systems monitoring and  network inference. 
\end{IEEEbiographynophoto}

\begin{IEEEbiographynophoto}
{Joe H. Chow} (F'92) received the M.S. and Ph.D. degrees from the University of Illinois, Urbana-Champaign, Urbana, IL, USA.

After working in the General Electric power system business in Schenectady, NY, USA, he joined Rensselaer Polytechnic Institute, Troy, NY, USA, in 1987, where he is a Professor of Electrical, Computer, and Systems Engineering. His research interests include multivariable control, power system dynamics and
control, FACTS controllers, and synchronized phasor data.
\end{IEEEbiographynophoto}

 \begin{IEEEbiographynophoto}
 {Scott G. Ghiocel} (S'08) received the Ph.D. degree in electrical engineering from Rensselaer Polytechnic Institute, Troy, NY, USA, in 2013.
  
 He is a technical consultant at Exponent. His research interests include synchrophasor measurements, voltage stability, and power system dynamics.
 \end{IEEEbiographynophoto}
 
\begin{IEEEbiographynophoto}
{Bruce Fardanesh}(F'13) received his  Doctor of Engineering degree in Electrical Engineering from Cleveland State University in 1985.  
  
He joined New York Power Authority in 1991, where he is the Chief Electrical Engineer. His research areas of interest are power system analysis, modeling, dynamics, operation, and control. 
\end{IEEEbiographynophoto}
 
\begin{IEEEbiographynophoto}
{George Stefopoulos} (M'08) received his Ph.D. degree in Electrical Engineering from the Georgia Institute of Technology in   2009. 
 
He is a Research and Technology Development Engineer with the New York Power Authority. His research interests include  power system state estimation, synchrophasor technology applications, and modeling and simulation of power systems. 
\end{IEEEbiographynophoto}  
  
\begin{IEEEbiographynophoto}
{Michael Razanousky} received the B.S. degree in electric power engineering from Rensselaer Polytechnic Institute, Troy, NY, USA, in 1989 and the M.S. degree from the University of Albany, Albany, NY, USA, in 2005.
   
He is presently a project manager at the New York State Energy and Research Development Authority (NYSERDA).
\end{IEEEbiographynophoto}

\end{document}